\documentclass[12pt,a4paper]{article}
\usepackage{amsthm,amsmath,amsfonts,amssymb,amsxtra,bookmark}

\theoremstyle{plain}
\newtheorem{theorem}{Theorem}[section]
\newtheorem{lemma}[theorem]{Lemma}

\newtheorem{conjecture}[theorem]{Conjecture}

\theoremstyle{definition}
\newtheorem*{definition}{Definition}
\newtheorem{example}[theorem]{Example}

\theoremstyle{remark}
\newtheorem*{remark}{Remark}

\DeclareMathOperator{\Tr}{Tr}
\DeclareMathOperator{\Ker}{Ker}

\DeclareMathOperator{\Span}{Span}

\def\geqslant{\ge}
\def\leqslant{\le}
\def\bq{\begin{eqnarray}}
\def\eq{\end{eqnarray}}
\def\bqq{\begin{eqnarray*}}
\def\eqq{\end{eqnarray*}}
\def\nn{\nonumber}
\def\minus {\backslash}
\def\eps{\varepsilon}
\def\wto{\rightharpoonup}

\def\cV {\mathcal{V}}
\def\R {\mathbb{R}}
\def\C {\mathbb{C}}
\def\N {\mathcal{N}}
\def\U {\mathbb{U}}
\def\cJ {\mathcal{J}}
\def\cS {\mathcal{S}}

\def\E {\mathcal{E}}

\def\cA{\mathcal{A}}

\def\F {\mathcal{F}}
\def\B {\mathcal{B}}
\def\H{\mathcal{H}}
\def\h{\mathfrak{h}}
\def\fh{\mathfrak{h}}
\def\V {\mathcal{V}}
\def\cV {\mathcal{V}}
\def\R {\mathbb{R}}
\def\C {\mathbb{C}}
\def\N {\mathcal{N}}
\def\U {\mathbb{U}}
\def\J {\mathcal{J}}
\def\cJ {\mathcal{J}}
\def\S {\mathcal{S}}
\def\cS {\mathcal{S}}
\def\G {\mathcal{G}}
\def\E {\mathcal{E}}
\def\A{\mathcal{A}}
\def\cA{\mathcal{A}}

\title{\bf \Large Bogoliubov theory and bosonic atoms}
\author{Phan Thanh Nam \\\\
\small Department of Mathematical Sciences, University of Copenhagen,\\
\small Universitetsparken 5, 2100 Copenhagen, Denmark. E-mail: ptnam@math.ku.dk}

\begin{document}
\date{{}}
\maketitle

\begin{abstract} 
We formulate the Bogoliubov variational principle in a mathematical framework similar to the generalized Hartree-Fock theory. Then we analyze the Bogoliubov theory for bosonic atoms in details. We discuss heuristically why the Bogoliubov energy should give the first correction to the leading energy of large bosonic atoms. 
\end{abstract}

\tableofcontents
\addcontentsline{toc}{section}{Contents}

\section{Bogoliubov theory}

In this section we formulate the Bogoliubov variational principle in the same spirit of the generalized Hartree-Fock theory \cite{BLS94}. Our formulation bases on the earlier discussions in \cite{So06,So07}.

\subsection{One-body density matrices}

We start by introducing some conventional notations. Let $\h$ be a complex separable Hilbert space with the inner product $(.,.)$ which is linear in the second variable and anti-linear in the first. Let $\h_N:=\bigotimes _{\text{sym}}^N \h$ be the symmetric tensor product space of $N$ particles and let $\F=\F(\h):=\bigoplus_{N=0}^\infty \h_N$ be the bosonic {\it Fock space}.
 
Let $\B(\F)$ be the space of linear bounded operators on $\F$. Any  {\it quantum mechanical state} ({\it state} for short) $\rho:\B(\F)\to \C$ is identified with a positive semi-definite trace class operator $P$ on $\F$ with $\Tr(P)=1$ in such a way that 
$$\rho(B)=\text{Tr}(BP)~~\text{for all}~B\in \B(\F).$$
For example, a {\it pure state} is a state corresponding to the one-dimensional projection $\left| {\Psi} \right\rangle \left\langle {\Psi} \right|$ of a unit vector $\Psi\in \F$, and a {\it Gibbs state} is a state corresponding to $\Tr(\exp(-H))^{-1}\exp (-H)$ for some Hamiltonian $H:\F\to \F$ such that $\exp (-H)$ is trace class.

The {\it dual space} $\h^*$ can be identified to $\h$ by the {\it anti-unitary } $J:\h\to \h^*$, $$J(x)(y)=(x,y)_\h,~~\text{for all}~x,y\in \h.$$
It is convenient to introduce the {\it generalized annihilation} and {\it creation} operators on $\h\oplus \h^*$ by
\bqq A(f\oplus Jg)&=&a(f)+a^*(g),\hfill\\
A^*(f\oplus Jg)&=&a^*(f)+a(g),~~\text{for all}~f,g\in \h
\eqq
where $a(f)$ and $a^{*}(f)$ are the usual annihilation and creation operators. Note that if we denote 
\[
\S = \left( \begin{gathered}
  1~{\text ~~~}0 \hfill \\
  0~~-1 \hfill \\ 
\end{gathered}  \right)~~, \J = \left( {\begin{array}{*{20}c}
   0 & {J^* }  \\
   J & 0  \\

 \end{array} } \right),
\]
then we have the {\it conjugate relation} and the {\it canonical commutation relation} (CCR)
\bqq \label{eq:conjugate-CCR}
A^*( F_1)=A(\J F_1)~~, \left[ {A(F_1 ),A^* (F_2 )} \right] = (F_1 ,\S F_2 )~~\text{for all}~ F_1,F_2\in  \h\oplus \h^*
\eqq 
where $[X,Y]=XY-YX$.

Now we can define the {\it one-particle density matrix} ({\it 1-pdm} for short) $\Gamma:\h\oplus \h^* \to \h\oplus \h^*$ of a state $\rho$ by 
\[
(F_1 ,\Gamma F_2 ) = \rho (A^ *  (F_2 )A(F_1 ))~~~~\text{for all}~F_1,F_2\in \h\oplus \h^*.
\]
Such a 1-pdm may be also written as
\bq \label{eq:1pdm}
\Gamma  = \left( {\begin{array}{*{20}c}
   \gamma  & \alpha   \\
   {J\alpha J } & {1 + J\gamma J^*}  \\

 \end{array} } \right)
\eq
where $\gamma:\h\to \h$ and $\alpha:\h^*\to \h$ are {linear bounded} operators defined by 
\bqq
  (f,\gamma g)  = \rho (a^* (g)a(f)), ~  (f,\alpha Jg) = \rho (a(g)a(f))~~\text{for all}~f,g\in \h.
\eqq
It is obvious that any 1-pdm is {\it positive semi-definite}. The following lemma expresses the condition $\Gamma\ge 0$ in terms of $\gamma$ and $\alpha$. Its proof is provided in the Appendix.

\begin{lemma}\label{le:relation-gamma-alpha} Let $\Gamma$ be of the form (\ref{eq:1pdm}). Then $\Gamma\ge 0$ if and only if $\gamma\ge 0$, $\alpha^*=J\alpha J$ and 
\bq \label{eq:relation-gamma-alpha}
\gamma  \geqslant \alpha J (1 + \gamma )^{ - 1} J^* \alpha ^* .
\eq
\end{lemma}

\begin{remark} The fermionic analogue of the inequality (\ref{eq:relation-gamma-alpha}) is 
$\alpha \alpha ^*\le \gamma (1 - \gamma )$ \cite{BLS94}. We do not know if (\ref{eq:relation-gamma-alpha}) can be reduced to $\alpha \alpha ^*\le \gamma (1 + \gamma )$ or not.
\end{remark} 

Of primary physical interest are the states  with {\it finite particle number expectation}. Recall the {\it particle number operator}
$$\N:=\sum_{N=0}^\infty N 1_{\h_N}= \sum\limits_n {a^* (u_n )a(u_n )}$$
for any orthonormal basis $\{u_n\}_{n=1}^{\infty}$ for $\h$. It is straightforward to see that if a state $\rho$ has the 1-pdm of the form (\ref{eq:1pdm}) then 
$$\rho(\N)=\Tr(\gamma).$$
Hence $\rho$ has finite particle number expectation if and only if $\gamma$ is {trace class}.

\subsection{Bogoliubov transformations}

\begin{definition}[Bogoliubov transformations] A bosonic {\it Bogoliubov transformation} is a {linear} bounded isomorphism $\V:\h\oplus \h^*\to \h\oplus \h^*$ satisfying
$$\J\V\J=\V~~\text{and}~\V^*\S\V=\S.$$
\end{definition}

These conditions ensure that the Bogoliubov transformations preserve the conjugate relation and the canonical commutation relation, namely
\[
A^*(\V F_1)=A(\V\J F_1)~~\text{and}~\left[ {A(\V F_1 ),A^* (\V F_2 )} \right] = (F_1 ,\S F_2 ),~~\forall F_1 ,F_2\in  \h\oplus \h^*.
\]

The Bogoliubov transformations form a subgroup of the isomorphisms in  $\h\oplus \h^*$; in particular, if $\V$ is a Bogoliubov transformation then $\V^{-1}$ and $\V^*$ are also Bogoliubov transformations. Note that any mapping $\V$ satisfying $\J\V\J=\V$ must have the form
\bq
\V = \left( {\begin{array}{*{20}c}
   U & V  \\
   {JVJ} & {JUJ^* }  \\

 \end{array} } \right)\label{eq:Bogoliubovmap} 
\eq
for some linear operators $U:\h\to \h$, $V:\h^*\to \h$. 

We say that a Bogoliubov transformation $\V$ is {\it unitarily implementable} if it is {implemented} by a unitary mapping $\U_\V:\F\to \F$, namely
\bq \label{eq:Bogoliubov-unitary-mapping}
A(\V F) = \U_\V A(F)\U_\V^* ~~\text{for all}~F\in \h\oplus \h^*.
\eq
The following result determines whenever a Bogoliubov transformation is unitarily implementable. This result is well-known and we provide its proof in the Appendix for the reader's convenience. For the fermionic analogue, see \cite{BLS94} (Theorem 2.2) .  

\begin{theorem}[Unitarily implementable Bogoliubov transformations]\label{thm:unitary-implementation} A Bogoliubov transformation $\V:\h\oplus\h^*\to\h\oplus\h^*$ of the form (\ref{eq:Bogoliubovmap}) is unitarily implementable if and only if the Shale-Stinespring condition $\Tr_\h(VV^*)<\infty$ holds. 
\end{theorem}

Unlike to the fermionic case \cite{BLS94}, the bosonic Bogoliubov transformations are {\it not} unitary mappings on $\fh\oplus\fh^*$. However, we can still use the Bogoliubov transformations to diagonalize some certain operators on $\fh\oplus\fh^*$. Of our particular interest is the diagonalization of the 1-pdm's.   
   
\begin{theorem}[Diagonalization 1-dpm's by Bogoliubov transformations]\label{thm:diagonalizing-Gamma} If $\Gamma$ has the form (\ref{eq:1pdm}) with $\Gamma\ge 0$ and $\text{Tr}(\gamma)<\infty$ then for an arbitrary orthonormal basis $\{u_n\}$ for $\h$, there is a unitarily implementable Bogolubov transformation $\V:\h\oplus \h^*\to \h\oplus \h^*$ diagonalizing $\Gamma$ in in the basis $u_1\oplus 0$, $u_2\oplus 0$, ..., $0\oplus Ju_1$, $0\oplus Ju_2$, ..., namely 
\bq
\V^* \Gamma \V = \left( {\begin{array}{*{20}c}
   {\lambda _1 } & {} & {} & {} & {} & {}  \\
   {} & {\lambda _2 } & {} & {} & 0 & {}  \\
   {} & {} &  \ddots  & {} & {} & {}  \\
   {} & {} & {} & {1 + \lambda _1 } & {} & {}  \\
   {} & 0 & {} & {} & {1 + \lambda _2 } & {}  \\
   {} & {} & {} & {} & {} &  \ddots   \\

 \end{array} } \right),\label{eq:diagonalizingGamma}
\eq
\end{theorem}

\begin{remark}  The finite-dimensional case is Theorem 9.8 in \cite{So07}. See \cite{BLS94} (the proof of Theorem 2.3) for the fermionic analogue.
\end{remark}

To prove Theorem \ref{thm:diagonalizing-Gamma}, we start with a simple diagonalization lemma. This is a generalization to infinity dimensions of Lemma 9.6 in \cite{So07}.

\begin{lemma}
\label{lm:bogoliubov-diag} Let $\A$ be a positive definite operator on $\fh\oplus\fh^*$ such that $\cJ\cA\cJ=\cA$ and $\S \cA$ admits an eigenbasis 
on $\fh\oplus \fh^*$. Then for any orthonormal basis $u_1,u_2,...$ for $\fh$, there exists a Bogoliubov transformation $\cV$ such that the operator $\cV ^*\cA \cV$ has eigenvectors of the form $\{u_n\oplus 0\}\cup \{0\oplus Ju_n\}$.
  \end{lemma}
  
\begin{remark} In this result the Bogoliubov transformation $\V$ needs {\it not} be unitarily implementable.
\end{remark}
  
\begin{proof} 1. Let $\{u_i\}$ be an orthonormal basis for $\h$. We shall define the Bogoliubov transformation $\cV$ by  
$$
\cV (u_i\oplus0)=v_{i},~ \cV (0\oplus
    Ju_i)=\widetilde{v}_i,
$$
where $\{v_i\} \cup \{\widetilde{v}_i\}$ is an eigenbasis of $\cS \cA$ such that 

\begin{itemize}

\item[(i)] $(v_i,\cS v_j)=\delta_{ij},  (\widetilde{v}_i,\cS \widetilde{v}_j)=-\delta_{ij}$ 
and $(v_i,\cS \widetilde{v}_j)=0$ for all $i,j=1,2,\ldots$ \label{item:symplecticortho}
    
\item[(ii)] $\cJ v_j =\widetilde{v}_{j}$ for all $j=1,2,\ldots$ \label{item:Jcond}
\end{itemize} 

2. Let $v_1$ be a normalized eigenvector of $\cS \cA$ with eigenvalue $\lambda_1$. Using $\cA v_1 =\lambda_1 \cS v_1$ we find that
$$
 (v_1,\cA v_1)=\lambda_1(v_1,\cS v_1).
$$
Since $\cA$ is positive definite and $\cS$ is Hermitian, both of $\lambda_1$ and $(v_1,Sv_1)$ must be real and non-zero. Therefore, we can normalized $v_1$ in such a way that $(v_1,\cS v_1)\in \{\pm 1\}.$

Defining $\widetilde{v}_{1}=\cJ v_1$ and using $\J \A \J=\A$ we have that 
    $$
    \cS\cA \widetilde{v}_{1}=\cS\cJ\cA v_1=-\cJ\cS \cA v_1=-\cJ \lambda_1 v_1 = -\lambda_1 v_{M+1},
    $$
where we have used that $\lambda_1$ is real and that $\cJ\cS=-\cS\cJ$. Thus $\widetilde{v}_{1}$ is an eigenvector of $\S \A$ with the eigenvalue $\widetilde{\lambda}_1=-\lambda_1$.  
    
    Since $\lambda_1\ne 0$, $\widetilde{\lambda}_1$ and $\lambda_1$ must be different. On the other hand, 
    $$
   \widetilde{\lambda}_1 (v_1,\cS \widetilde{v}_1)=(v_1,\cA
    \widetilde{\lambda}_1)=(\cA v_1,\widetilde{v}_1)=\lambda_{1}(v_1,\cS \widetilde{v}_1).
    $$
Thus $(v_1,\cS \widetilde{v}_1)=0$. Moreover, since 
    $$
    (\widetilde{v}_1,\cS \widetilde{v}_1)=(\cJ v_1,\cS\cJ v_1)=(\cJ v_1,-\cJ\cS v_1)=
    -(\cS v_1,v_1)=-(v_1,\cS v_1),
    $$
    by interchanging $v_1$ and $\widetilde{v}_1$ if necessary we can assume that $(v_1,\cS v_1)=1$ and $(\widetilde{v}_1,\cS \widetilde{v}_1)=-1$. 

3. Let $V=\Span\{v_1,\widetilde{v}_1\}$ and $W=(\S V)^{\bot }=\S (V^{\bot})$. We shall show that  
$$ \fh\oplus \fh^*= V \oplus W.$$

Indeed, if $a \in V \cap W$ then $a
\in V=\Span\{v_1,\widetilde{v}_1\}$ and $(a,Sv)=0$ for all $v
\in V$. Because $(v_1,\cS v_1)=1$, $(\widetilde{v}_1,\cS \widetilde{v}_1)=-1$ and $(v_1,\cS \widetilde{v}_1)=0$, we must have $a=0$. Thus $V \cap W = \{ 0\}$. 

On the other hand, if $a \in (V \oplus W)^{\bot}\subset V^{\bot}\cap W^{\bot}$ then $\cS a\in S (V^{\bot}) \cap \cS (W^{\bot}) =W\cap V=\{0\}$, and hence $a=0$. Therefore,  $(V \oplus W)^{\bot} = \{ 0\}$.

Moreover, since $V$ is finite dimensional and $W$ is closed, the direct sum space  $V \oplus W$ is a closed subspace of $ \fh\oplus \fh^*$. Thus $ \fh\oplus \fh^*= V \oplus W.$

4. We prove that $\cS\cA$ maps $W$ into itself. Indeed, using $V = \cS \cA V$ we have $
W \bot \cS V = \cS(\cS \cA V) = \cA V$. Since $A$ is symmetric, we get $\cA W \bot V$, and hence $\cS \cA W \bot \cS V$. Thus $\cS \cA W \subset (\cS V)^{\bot }=W$.

Because $\cS \cA$ admits an eigenbasis on $\fh\oplus \fh^*= V \oplus W$ and $\cS \cA$ leaves $V$ and $W$ invariant, $\cS \cA$ also admits an eigenbasis on $W$. We then can restrict $\cS \cA$ on $W$ and conclude the desired result by an induction argument. 
\end{proof}

Next, we show that $\Gamma+\frac{1}{2}\S$ satisfies all assumptions on $\cA$ in Lemma \ref{lm:bogoliubov-diag}.   

\begin{lemma} \label{le:properties-Gamma1} Let $\Gamma$ be of the form (\ref{eq:1pdm}) with $\Gamma\ge 0$ and $\text{Tr}(\gamma)<\infty$ and let $\Gamma_1:=\Gamma+\frac{1}{2}\S$. Then $\Gamma_1$ is positive definite on $\fh\oplus\fh^*$; moreover, $\cJ\Gamma_1\cJ=\Gamma_1$ and $\S \Gamma_1$ admits an eigenbasis on $\fh\oplus \fh^*$.
\end{lemma}

\begin{proof} 1. It is straightforward to check that $\J \Gamma_1 \J =\Gamma_1$. We now prove that $\Gamma_1$ is positive definite. 

First at all, it follows from $\Gamma\ge 0$ that 
\[
\left\langle {f \oplus Jg,(\Gamma  + \S)f \oplus Jg} \right\rangle  = \left\langle {g \oplus Jf,\Gamma (g \oplus Jf)} \right\rangle \ge 0,
\]
namely  $\Gamma+\S\ge 0$. Thus 
$$\Gamma_1=\Gamma+\frac{1}{2}\S=\frac{1}{2}[\Gamma+(\Gamma+\S)]\ge 0.$$

Next, we check that $\Gamma_1$ is injective. Assume that there exists 
$\varphi\in \Ker(\Gamma_1)\minus \{0\}$. Then since $\J$ and $\Gamma_1$ commute, we have $\J\varphi\in \Ker(\Gamma_1)\minus \{0\}$. Because $\J$ leaves the subspace $\Span \{\varphi,\J \varphi\}\subset \Ker(\Gamma_1)$ invariant, $\J$ must have a non-trivial fixed point in this subspace. Thus there exists $f\in \h \backslash\{0\}$ such that $\Gamma_1(f\oplus Jf)=0$. Using this equation we find that
\bqq
\left\langle {f \oplus J(tf),\Gamma (f \oplus J(tf))} \right\rangle
   &=& (f,\gamma f) + t^2 (f,(1 + \gamma )f) - t\left( {f,(2\gamma  + 1)f} \right) \hfill \\
   &=& (t - 1)^2 (f,\gamma f) + (t^2  - t)\left\| f \right\|^2<0
\eqq
for some $t<1$ and near $1$ sufficiently. However, it is contrary to $\Gamma\ge 0$. Thus $\Gamma_1$ must be injective.

To see that $\Gamma_1$ is positive definite we can introduce $\Gamma_1^{1/2}$, the unique positive semi-definite square root $\Gamma_1^{1/2}$ on $\h\oplus \h^*$. Since $\Gamma_1$ is injective, $\Gamma_1^{1/2}$ is also injective, and hence 
$$(\varphi, \Gamma_1 \varphi)=||\Gamma_1^{1/2}v||^2>0~~\text{for all}~\varphi\ne 0.$$

2. We show that $\S \Gamma_1$ has an eigenbasis on $\h\oplus \h^*$. Although $\S\Gamma_1$ is not a Hermitian, we may associate it with the Hermitian $C=\Gamma_1^{1/2}\S\Gamma_1^{1/2}$. 

We can see that $C$ has an orthonormal eigenbasis for $\h\oplus \h^*$. Indeed, it is straightforward to see that 
\bqq
C^2  &=& \Gamma _1^{1/2} (S\Gamma S)\Gamma _1^{1/2}  = \Gamma _1^{1/2} \left[ {\left( {\begin{array}{*{20}c}
   \gamma  & { - \alpha }  \\
   { - \alpha ^* } & {J\gamma J^* }  \\

 \end{array} } \right) + \frac{1}
{2}\operatorname{I} } \right]\Gamma _1^{1/2}  \hfill \\
   &=& \Gamma _1^{1/2} \left( {\begin{array}{*{20}c}
   \gamma  & { - \alpha }  \\
   { - \alpha ^* } & {J\gamma J^* }  \\

 \end{array} } \right)\Gamma _1^{1/2}  + \frac{1}
{2}\Gamma _1  \hfill \\
   &=& \Gamma _1^{1/2} \left( {\begin{array}{*{20}c}
   \gamma  & { - \alpha }  \\
   { - \alpha ^* } & {J\gamma J^* }  \\

 \end{array} } \right)\Gamma _1^{1/2}  + \frac{1}
{2}\left( {\begin{array}{*{20}c}
   \gamma  & \alpha   \\
   {\alpha ^* } & {J\gamma J^* }  \\

 \end{array} } \right) + \frac{1}
{4}\operatorname{I} .
\eqq
Because $\gamma$ is trace class, $\alpha \alpha^*$ is also trace class due to inequality (\ref{eq:relation-gamma-alpha}). Thus $(C^2-\frac{1}{4}\operatorname{I})$ is a self-adjoint Hilbert-Schmidt operator, and hence it has an orthonormal eigenbasis on $\h\oplus \h^*$. Therefore, $C^2$ has an orthonormal eigenbasis. Note that if $\varphi$ is an eigenvector of $C^2$ then $C\varphi$ is also an eigenvector of $C^2$ with the same eigenvalue. Because $C$ maps the subspace $\Span \{\varphi,C\varphi \}$ into itself, we can diagonalize to obtain an orthonormal eigenbasis of $C$ on this subspace. By induction, we get an orthonormal eigenbasis of $C$ on $\h\oplus \h^*$. 

Now note that if $\varphi$ is an eigenvector of $C$ then $\S\Gamma_1^{1/2}\varphi$ is an eigenvector of $\S\Gamma_1$ with the same eigenvalue since 
\[
S\Gamma _1 (S\Gamma _1^{1/2} \varphi ) = S\Gamma _1^{1/2} (\Gamma _1^{1/2} S\Gamma _1^{1/2} )\varphi  = S\Gamma _1^{1/2} (C \varphi).
\]
Moreover, because both of $S$ and $\Gamma_1^{1/2}$ are injective, $\S\Gamma_1^{1/2}$ maps a basis on $\h\oplus \h^*$ to another basis. In particular, $\S\Gamma_1^{1/2}$ maps an eigenbasis of $C$ to an eigenbasis of $\S\Gamma_1$. 
\end{proof}

Now we can prove Theorem \ref{thm:diagonalizing-Gamma} similarly to Theorem 2.3 in \cite{BLS94}). 

\begin{proof}[Proof of Theorem \ref{thm:diagonalizing-Gamma}] 1. Apply Lemma \ref{lm:bogoliubov-diag} with $\A=\Gamma_1:=\Gamma+\frac{1}{2}\S$, we can find a Bogoliubov transformation $\V$ on $\h \oplus \h^*$ such that, with respect to the orthonormal basis $\{u_n\oplus 0\}\cup \{0\oplus Ju_n\}$, 
\bqq \label{eq:diagonalizingGamma1}
\V^* \Gamma _1 \V = \left( {\begin{array}{*{20}c}
   {\lambda _1  + \frac{1}
{2}} & {} & {} & {} & {} & {}  \\
   {} & {\lambda _2  + \frac{1}
{2}} & {} & {} & 0 & {}  \\
   {} & {} &  \ddots  & {} & {} & {}  \\
   {} & {} & {} & {\lambda _1  + \frac{1}
{2}} & {} & {}  \\
   {} & 0 & {} & {} & {\lambda _2  +\frac{1}
{2}} & {}  \\
   {} & {} & {} & {} & {} &  \ddots   \\

 \end{array} } \right),
\eqq
which is equivalent to (\ref{eq:diagonalizingGamma}). 

We claim that in (\ref{eq:diagonalizingGamma}) we must have $\lambda_n\ge 0$ and $\sum_n \lambda_n<\infty$. It follows from (\ref{eq:diagonalizingGamma}) and $\V^* \Gamma\V\ge 0$ that $\lambda_n\ge 0$. In order to prove the boundedness $\sum_n \lambda_n<\infty$ we note that 
\[
\Gamma S(\Gamma  + S) = \left( {\begin{array}{*{20}c}
   {\gamma (\gamma  + 1) - \alpha \alpha ^* } & {\gamma \alpha  - \alpha J\gamma J^* }  \\
   {\alpha ^* \gamma  - J\gamma J^* \alpha ^* } & {\alpha ^* \alpha  - J\gamma (\gamma  + 1)J^* }  \\

 \end{array} } \right)
\]
is a self-adjoint trace class operator. Using the diagonal form  
\bqq \label{GammaSGamma+Gamma}
  &~&\V^* \Gamma \S(\Gamma  + \S)\V= \left[ {\V^* \Gamma\V} \right]\S\left[ {\V^* (\Gamma  + \S)\V} \right]\nn \hfill \\
   &=& \left( {\begin{array}{*{20}c}
   {\lambda _1 (\lambda _1  + 1)} & {} & {} & {} & {} & {}  \\
   {} & {\lambda _2 (\lambda _1  + 1)} & {} & {} & 0 & {}  \\
   {} & {} &  \ddots  & {} & {} & {}  \\
   {} & {} & {} & { - \lambda _1 (\lambda _1  + 1)} & {} & {}  \\
   {} & 0 & {} & {} & { - \lambda _2 (\lambda _1  + 1)} & {}  \\
   {} & {} & {} & {} & {} &  \ddots   \\

 \end{array} } \right)
\eqq
we conclude that $\sum_n \lambda_n(\lambda_n+1)<\infty$, which is equivalent to $\sum_n \lambda_n<\infty$.

2. Finally we show that the Bogoliubov transformation $\V$ constructed above is unitarily implementable. Assume $\V$ has the form (\ref{eq:Bogoliubovmap}). Then by Theorem \ref{thm:unitary-implementation}, it suffices to prove that $VV^*$ is a trace class operator on $\h$. 

It follows from the representation (\ref{eq:diagonalizingGamma}) that the upper left block of $\V^* \Gamma \V$ is a positive semi-definite trace class operator on $\h$. By direct computation, we can see that the upper left block of
\[
\V^* \Gamma \V = \left( {\begin{array}{*{20}c}
   {U^* } & {J^* V^* J^* }  \\
   {V^* } & {JU^* J^* }  \\

 \end{array} } \right)\left( {\begin{array}{*{20}c}
   \gamma  & \alpha   \\
   {\alpha ^* } & {1 + J\gamma J^* }  \\

 \end{array} } \right)\left( {\begin{array}{*{20}c}
   U & V  \\
   {JVJ} & {JUJ^* }  \\

 \end{array} } \right),
\]
is
$$U^*\gamma U+J^*V^*J^*U+U^*\alpha JVJ+J^*V^*VJ+J^*V^*\gamma VJ.$$
Because $\gamma$ is trace class, we have $U^*\gamma U$ and $J^*V^*\gamma VJ$ are trace class. Thus $J^*V^*J^*U+U^*\alpha JVJ+J^*V^*VJ$ is trace class. 

Moreover, using the Cauchy-Schwarz inequality 
$$|\operatorname{Tr}(XY+Y^*X^*)|\le 2(\operatorname{Tr} (XX^*))^{1/2}(\operatorname{Tr}(YY^*))^{1/2}$$
we find that 
\bqq
  \infty  &>& \operatorname{Tr} \left[ {U^* \alpha JVJ + J^* V^* J^* \alpha ^* U + J^* V^* VJ} \right] \hfill \\
   &=& \operatorname{Tr} \left[ {(U^* \alpha J)(VJ)+(VJ )^*(U^* \alpha J)^*} \right] + \operatorname{Tr} (VV^* ) \hfill \\
   &\geqslant&  - 2(\operatorname{Tr} (U^* \alpha \alpha ^* U^* ))^{1/2} (\operatorname{Tr} (VV^* ))^{1/2}  + \operatorname{Tr} (VV^*).
\eqq
Note that $\operatorname{Tr} (U^* \alpha \alpha ^* U^* )<\infty$ because $\alpha \alpha^*$ is trace class. Thus $\operatorname{Tr} (VV^*)<\infty$. 
\end{proof}

\subsection{Quasi-free states and quadratic Hamiltonians}

\begin{definition}[Quasi-free states] A {\it quasi-free state} $\rho$ is a state satisfying Wick's Theorem, namely 
\bq\label{eq:Wick-odd}
\rho[A(F_1)...A (F_{2m-1})]=0~~\text{for all}~m\ge 1
\eq
and 
\bq \label{eq:Wick-even}
&~&\rho[A (F_1)...A (F_{2m})]=\sum_{\sigma \in P_{2m}}  \rho[A (F_{\sigma(1)}) A (F_{\sigma(2)})]...\rho[A (F_{\sigma(2m-1)}) A (F_{\sigma(2m)})]
\eq
where $P_{2m}$ is the set of pairings
    \begin{eqnarray*}
      P_{2m}=\{\sigma\in S_{2m}\ |\ \sigma(2j-1)<\sigma(2j+1),
      j=1,\ldots,m-1,\quad &&\\
      \sigma(2j-1)< \sigma(2j),\ j=1,\ldots,m\}.&&
    \end{eqnarray*}
\end{definition}    
A crucial point is that we have {\it one-to-one correspondence} between the set of quasi-free states with finite particle numbers and the set of 1-pdm's. If a quasi-free state is a pure state, namely a one-dimensional projection on the Fock space, we call it a {\it quasi-free pure state}. 

\begin{theorem}[Quasi-free states and quasi-free pure states]\label{thm:quasi-free-state} \text{}\begin{itemize}
\item [(i)] Any operator $\Gamma:\h\oplus \h^*\to\h\oplus \h^*$ of the form (\ref{eq:1pdm}) satisfying $\Gamma\ge 0$ and ${\rm Tr}(\gamma)<\infty$ is the 1-pdm of a quasi-free state with finite particle number expectation.

\item[(ii)] A pure state $\left| \Psi  \right\rangle \left\langle \Psi  \right|$ with finite particle number expectation is a quasi-free state if and only if $\Psi= \U _ \V \left| 0 \right\rangle $ for some Bogoliubov unitary mapping $\U _ \V$ as in (\ref{eq:Bogoliubov-unitary-mapping}).

Moreover, any operator $\Gamma:\h\oplus \h^*\to\h\oplus \h^*$ of the form (\ref{eq:1pdm}) satisfying $\Gamma\ge 0$ and ${\rm Tr}(\gamma)<\infty$ is the 1-pdm of a quasi-free pure state if and only if $\Gamma\S\Gamma =-\Gamma$. 
\end{itemize}
\end{theorem}

\begin{remark} The characterization of quasi-free pure states were already proved in \cite{So07} (with a different proof). For the fermionic analogues see \cite{BLS94} (Theorem 2.3 and Theorem 2.6). 
\end{remark}

\begin{proof} {\bf (i)} Note that the set of quasi-free states is {\it invariant} under Bogoliubov unitary mappings. Indeed, if the Bogoliubov transformation $\V$ is implemented by the unitary mapping $\U_\V:\F\to \F$ as in (\ref{eq:Bogoliubov-unitary-mapping}) and $\Gamma$ is the 1-pdm of a quasi-free state $\rho$ then $\V^*\Gamma\V$ is the 1-pdm of the quasi-free state $\rho_{\V^*\Gamma\V}$ defined by
$$\rho_{\V^*\Gamma\V}(B):=\rho(\U_\V B \U_\V^*)~~\text{for all}~B\in \B(\F).$$

Therefore, due to the diagonalization result  in Theorem \ref{thm:diagonalizing-Gamma}, it remains to show that any operator of the form
\[
\Gamma  = \left( {\begin{array}{*{20}c}
   \xi  & 0  \\
   0 & {1 + J\xi J^* }  \\

 \end{array} } \right),
\]
where $\xi$ is a positive semi-definite trace class operator on $\h$, is indeed the 1-pdm of some quasi-free state. 

Because $\xi$ is trace class, it admits an orthogonal eigenbasis $\{u_i\}_{i=1}^\infty$ for $\h$ corresponding 
 to eigenvalues $\{\lambda_i\}_{i=1}^\infty$. Let $I=\{i\in \mathbb{N}|\lambda_i>0\}$. Then we may choose 
 $e_i\in (0,\infty)$ such that
\begin{eqnarray}
(1-\exp (-e_i))^{-1}=1+\lambda_i, \label{eq:e_i-lambda_i}~~i\in I.
\end{eqnarray}
Denote $a_i=a(u_i)$ for short. Let 
\bq \label{eq:quasi-free-state-form-G}
G= \Pi_0 \exp\left[ {-\sum_{i \in I} e_i a_i ^*a_i} \right]
\eq
where $\Pi_0$ is the orthogonal projection onto the subspace ${\rm Ker}[\sum_{i\notin I} a_i^*a_i]$. Similarly to the fermionic case (see Theorem 2.3 in \cite{BLS94}), it is straightforward to check that $\Gamma$ is the 1-pdm of the state $\rho=\Tr[G]^{-1}G$ and that $\rho$ is a quasi-free state. For the reader's convenience we provide this part of the proof in the Appendix. \\

{\bf (ii)} If $\Psi= \U _ \V \left| 0 \right\rangle$ for some  Bogoliubov unitary mapping $\U _ \V$ then using $\U_\V^*=\U_{\V^{-1}}$ and (\ref{eq:Bogoliubov-unitary-mapping}) we have $A(\V^{-1} F) = \U_\V^* A(F)\U_\V$. Therefore, 
\bqq
  \left\langle \Psi  \right|A(F_1 )A(F_2 )...A(F_n )\left| \Psi  \right\rangle & =& \left\langle 0 \right|\U_\V^* A(F_1 )\U_\V \U_\V A(F_2 )\U_\V ...\U_\V^* A(F_n )\U_\V \left| 0 \right\rangle  \hfill \\
   &=& \left\langle 0 \right|A(\V^{-1} F_1 )A(\V^{-1} F_2 )...A(\V^{-1} F_n )\left| 0 \right\rangle . 
 \eqq
 It is then obvious that the state $\left| \Psi  \right\rangle \left\langle \Psi  \right|$ satisfies equations (\ref{eq:Wick-odd})-(\ref{eq:Wick-even}) in Wick's Theorem, and hence it is a quasi-free state. 

Reversely, suppose that the pure state $\left| \Psi  \right\rangle \left\langle \Psi  \right|$ is a quasi-free state with finite particle number expectation. Then by the first statement in Theorem \ref{thm:quasi-free-state}, 
$$\left| \Psi  \right\rangle \left\langle \Psi  \right| = \Tr[G]^{-1}\U_\V G \U_\V^* $$
for some Bogoliubov unitary mapping $\U_\V$ and for some $G$ given by (\ref{eq:quasi-free-state-form-G}). On the other hand, the only rank-one operator $G$ of the form (\ref{eq:quasi-free-state-form-G}) is the vacuum projection $
\left| 0 \right\rangle \left\langle 0  \right|$ (namely $\xi=0$). Thus, up to a complex phase, $\Psi$ is equal to $\U_\V \left| 0 \right\rangle$.

Now we consider the 1-dpm's of quasi-free pure states.  Suppose that $\Psi$ is a quasi-free pure state with finite particle number expectation and its 1-dpm is $\Gamma$. Due to Theorem~\ref{thm:quasi-free-state}, there is a unitarily implementable Bogoliubov transformation $\V$ such that
    $$
    \cV^*\Gamma \cV = \left(\begin{array}{cc}0&0\\0&1
    \end{array}\right).
    $$
The identity $\Gamma\S\Gamma =-\Gamma$ follows from
   \bqq \V^*\Gamma\S\Gamma\V&=& (\V^*\Gamma\V)(\V^*\S\V)^{-1}(\V^*\Gamma\V)\\
    &=& \left(\begin{array}{cc}0&0\\0&1
    \end{array}\right) \S^{-1} \left(\begin{array}{cc}0&0\\0&1
    \end{array}\right) \\&=& - \left(\begin{array}{cc}0&0\\0&1
    \end{array}\right)=- \cV^*\Gamma \cV.
    \eqq

Reversely, let $\Gamma:\h\oplus \h^*\to\h\oplus \h^*$ be of the form (\ref{eq:1pdm}) such that $\Gamma\ge 0$, $\Tr(\gamma)<\infty$ and $\Gamma\S\Gamma =-\Gamma$. Then by Theorem~\ref{thm:quasi-free-state}, $\Gamma$ is the 1-dpm of a quasi-free state and there is a unitarily implementable Bogoliubov transformation $\V$ such that
$$
    \cV^*\Gamma \cV = \left(\begin{array}{cc}\xi&0\\0&1+J \xi J^*
    \end{array}\right)
 $$
for some positive semi-definite trace class operator $\xi$ on $\h$. The identity $\Gamma\S\Gamma =-\Gamma$ implies that 
$$ \left(\begin{array}{cc}\xi&0\\0&1+J \xi J^*
    \end{array}\right)\S \left(\begin{array}{cc}\xi&0\\0&1+J \xi J^*
    \end{array}\right)=-\left(\begin{array}{cc}\xi&0\\0&1+J \xi J^*
    \end{array}\right).$$
The only solution to this equation is $\xi=0$. Therefore, $\Gamma$ is the 1-dpm of a quasi-free pure state with finite particle number expectation. 
\end{proof}

One of the main motivation of considering the quasi-free pure states is that they minimize the quadratic Hamiltonians. 

\begin{definition}[Quadratic Hamiltonian] Let $\cA$ be a positive semi-definite operator on $\fh\oplus\fh^*$ and $\J\A\J=\A$.  The operator 
  $$
  H_\A=\sum_{i,j=1} (F_i,\cA F_j) A^*(F_i)A(F_j),
  $$
acting on $\F$ is called a {\it quadratic Hamiltonian} corresponding to $\cA$. Here $\{F_i\}_{i\ge 1}$ is an orthonormal basis for $\fh\oplus\fh^*$ (the sum is independent of the choice of $\{F_i\}_{i\ge 1}$).
\end{definition}

\begin{remark} \begin{itemize} \item[(i)] Alternatively, we can describe $H_\A$ by 
$$ (\Psi,H_\A \Psi)=\Tr[\A\Gamma_\Psi]~~\text{for all normalized vector}~\Psi\in \F,$$
where $\Gamma_\Psi$ is the 1-pdm of the pure state $\left| { \Psi } \right\rangle \left\langle { \Psi } \right|$.

\item[(ii)] The condition $\J\A\J=\A$ is just a conventional assumption since if this condition does not holds then we can consider $A'=\frac{1}{2}(\A+\J\A\J)$ which satisfies that $\J\A'\J=\A'$ and, formally, 
$$ H_{\A'}=H_{\A}+\frac{1}{2}\Tr[\A\S].$$
This formal formula makes sense when, for example, $\A$ is trace class.

\item[(ii)] As we shall see below, that $A\ge 0$ is the necessary and sufficient condition such that $H_{\A}$ is bounded from below. Moreover, in this case $H_{\A}\ge 0$.
\end{itemize}
\end{remark}

We are interested in the ground state energy of $H_{\A}$,  
\bq \label{eq:variational-quadratic-Hamiltonian} 
E(H_\A)  &:=& \inf  \{    \rho(H_\A) | \rho~\text{\rm is a state with}~\rho(\N)<\infty\} 
\eq

\begin{theorem}[Minimizing quadratic Hamiltonians]\label{thm:var-quad} Let $\A$, $H_\A$ and $E(H_\A)$ as above.
\begin{itemize}
\item[(i)] We have $E(H_\A)  = \inf  \{    \rho(H_\A) | \rho~\text{\rm is a quasi-free pure state}\} .$

\item[(ii)] If there is a unitarily implementable Bogoliubov transformation $\V_\A$ such that $\V_\A ^* \A\V_\A$ is diagonal then there is a quasi-free pure state $\rho_0$ such that $\rho_0(H_\A)=E(H_\A)$. Moreover, if $\A$ is positive definite then $\rho_0$ is unique.
\item[(iii)] If the variational problem (\ref{eq:variational-quadratic-Hamiltonian}) has a minimizer then $\A$ is diagonalized by a unitarily implementable Bogoliubov transformation $\V_\A$. Moreover, if $\Gamma$ is the 1-pdm of the minimizer then we have  
$$ \A\Gamma = -\A\S 1_{(-\infty,0)}[\A \S].$$
In particular, $\A \Gamma \S= \S \Gamma \A \le 0$. 
\end{itemize}
\end{theorem}

\begin{remark} \begin{itemize}

\item[(i)] The above statements (i) and (ii) already appeared in \cite{So07} in the finite-dimensional case (in this case $\A$ is always diagonalizable by Lemma \ref{lm:bogoliubov-diag}). 

\item[(ii)] If the operator $W$ is not self-adjoint but $U^{-1} W U$ is self-adjoint for some invertible operator $U$ then we can still define the projection $1_{(-\infty,0)}[W]$ by
$$1_{(-\infty,0)}[W]:=U 1_{(-\infty,0)}[U^{-1}W  U]U^{-1}.$$
It is easy to check that the definition is independent of the choice of $U$.  In particular, we can define 
$$1_{( - \infty ,0)} [\A\S] := (\V_\A^* )^{ - 1} 1_{( - \infty ,0)} [\V_\A^* \A\S(\V_\A^* )^{ - 1} ]\V_\A^* $$
where $\V_\A^* \A\S(\V_\A^* )^{ - 1} $ is self-adjoint.
\end{itemize}
\end{remark}  

\begin{proof} {\bf (i)} We show that for any state $\rho$ with finite number particle expectation, there is a quasi-free pure state $\widetilde \rho$ such that $ \widetilde \rho(H_\A)\le \rho(H_\A)$. 

By Theorem \ref{thm:diagonalizing-Gamma}, there is a unitarily implementable Bogoliubov transformation $\V$ such that 
\[
\Gamma  = \cV\left( {\begin{array}{*{20}c}
   \xi & 0  \\
   0 & 1 + J\xi J^*  \\
\end{array}} \right)  \cV^*
\]
where $\xi:\fh\to \fh$ is a positive semi-definite trace class operator. Thus
\begin{eqnarray*}
 \rho (H_ \A) = \Tr [\cA \Gamma ] = \Tr \left[ {\cV ^* \cA \cV   \left( {\begin{array}{*{20}c}
   \xi  & 0  \\
   0 & {1+J\xi J^* }  \\
\end{array}} \right)} \right] .
 \end{eqnarray*}
Because $\cV ^* \cA \cV $ commutes with $\J$, it has the block form
\[
\cV ^* \cA \cV   = \left( {\begin{array}{*{20}c}
   a & {b }  \\
   JbJ & { JaJ^* }  \\
\end{array}} \right)
\]
where $0\le a:\h\to \h$ and $b:\h^*\to \h$. Thus
\[
\rho (H_\cA) = \Tr \left[ {\left( {\begin{array}{*{20}c}
   a & {b }  \\
   JbJ & { JaJ^* }  \\
\end{array}} \right)\left( {\begin{array}{*{20}c}
   \xi  & 0  \\
   0 & {1 + J\xi J^* }  \\
\end{array}} \right)} \right] = 2\Tr [a\xi ] + \Tr [a] \ge \Tr[a].
\]

By Theorem \ref{thm:quasi-free-state}, there is a quasi-free pure state $ \widetilde \rho$ whose 1-pdm is 
\[
 \cV \left( {\begin{array}{*{20}c}
   0 & 0  \\
   0 & 1  \\
\end{array}} \right) \V ^* .
\]
It follows from the above discussion that $ \widetilde \rho(H_\A)\le \rho(H_\A)$.
\\

{\bf (ii)} Assume that $\A$ is diagonalized by the unitarily implementable Bogoliubov transformation $\V_\A$, namely 
\[
 \cA  = \cV _\cA ^* \left( {\begin{array}{*{20}c}
   d & 0  \\
   0 &  JdJ^*  \\
\end{array}} \right) \cV _\cA
\]
where $d:\fh\to \fh$ is positive semi-definite. For any state $\rho$ we have
$$
\rho(H_\cA)  = \Tr [\cA \Gamma] =  \Tr \left[ {\left( {\begin{array}{*{20}c}
   d & 0  \\
   0 & { J dJ^*}  \\
\end{array}} \right) \cV _\cA \Gamma  \cV _\cA^* } \right] 
$$
where $\Gamma$ is the 1-pdm of $\rho$. We may write $\cV _\cA \Gamma _\Psi  \cV _\cA^*$ in the block form 
\[
\cV _\cA \Gamma \cV _\cA ^*  = \left( {\begin{array}{*{20}c}
   \gamma  & \alpha   \\
   {\alpha ^* } & {1 + J \gamma J^*}  \\
\end{array}} \right)
\]
where $0\le \gamma:\fh\to \fh$ and $\alpha:\h^*\to \h$. Thus
\begin{eqnarray*}
 \rho(H_\cA)  = \Tr \left[ {\left( {\begin{array}{*{20}c}
   d & 0  \\
   0 & { J dJ^*}  \\
\end{array}} \right)\left( {\begin{array}{*{20}c}
   \gamma  & \alpha   \\
   {\alpha ^* } & {1 +J\gamma J^*}  \\
\end{array}} \right)} \right]=2\Tr [d\gamma ] + \Tr [d]
 \ge  \Tr [d] 
\end{eqnarray*}

Denote by $\rho_0$ the quasi-free pure state having the 1-dpm
\[
 \cV_A^{-1} \left( {\begin{array}{*{20}c}
   0 & 0  \\
   0 & 1  \\
\end{array}} \right)V_\A ^{*-1} = \S \cV_A^{*} \left( {\begin{array}{*{20}c}
   0 & 0  \\
   0 & 1  \\
\end{array}} \right)V_\A \S .
\]
Then $\rho_0(H_\A)=\Tr[d]$ and hence $\rho_0$ is a ground state of $H_\A$.

Moreover, if $\A$ is positive definite then $\Tr [d\gamma ] >0$ unless $\gamma=0$. Therefore, $\rho_0$ is the unique ground state of $H_\A$ among the quasi-free states.   
\\

{\bf (iii)} Assume that problem (\ref{eq:variational-quadratic-Hamiltonian}) has a minimizer and $\Gamma$ is the 1-dpm of the minimizer. 

1. We first prove that $\A\S$ and $\S\Gamma$ commute. Let $a$ be an arbitrary trace class operator on $\h\oplus \h^*$ such that $a=a^*=\J a \J$. It is straightforward to check that $\exp(i\eps H\S)$ is a Bogoliubov unitarily implementable transformation for any $\eps\in \R$.  Similarly to the variational argument for Hartree-Fock-Bogoliubov theory in \cite{LL01} (p. 284), we consider the trial states 
$$\Gamma _\varepsilon  : =\exp(-i\eps \S a) \Gamma \exp(i\eps a \S)=\Gamma+ \eps[i\Gamma a \S -i\S a \Gamma]+ O(\eps^2),~~\eps\in \R.$$
Since $\eps=0$ minimizes the functional $\eps\mapsto \Tr[\A(\Gamma_\eps-\Gamma)]$ we find that  
\[
0 = \frac{d}
{{d\varepsilon }}\operatorname{Tr} [\A(\Gamma _\varepsilon   - \Gamma )] = \operatorname{Tr} [aB]~~\text{with}~B:=i\S\A\Gamma  - i\Gamma \A\S.
\]
Note that $B=B^*=\J B \J$.

Now let $b$ be any trace class operator  on $\h\oplus \h^*$. Since $a:=b+b^*+\J b \J+\J b^* \J$ satisfies that $a=a^*=\J a \J$, we have 
\[
0=\Tr [aB]=4 \Re \Tr[bB].
\] 
By changing $b$ a complex phase, we conclude that $\Tr[bB]=0$ for any trace class operator $b$. This implies that $B=0$. Thus $\A \S \Gamma= \Gamma \A \S$, namely $[\A \S, \S \Gamma]=0$. 

2. Now let $\V_\Gamma$ be the unitarily implementable Bogoliubov  transformation such that 
\[
\V_\Gamma^* \Gamma \V_\Gamma = \left( {\begin{array}{*{20}c}
   \xi  & 0  \\
   0 & {1 + J\xi J^* }  \\

 \end{array} } \right)
\]
for some trace class operator $\xi\ge 0$ on $\h$. Thus the operator $\V_\Gamma^{-1} \S \Gamma \V_\Gamma=\S \V_\Gamma^* \Gamma \V_\Gamma$ leaves the spaces $\h\oplus 0$ and $0\oplus \h^*$ invariant . Since $\V_\Gamma^{-1} \A\S \V_\Gamma$ commutes with $\V_\Gamma^{-1} \S \Gamma \V_\Gamma$, it also leaves the spaces $\h\oplus 0$ and $0\oplus \h^*$ invariant. Moreover, since $\V_\Gamma^{-1} \A \S \V_\Gamma$ commutes with $\J$, it must have the form 
\[
\V_\Gamma^{-1} \A \S \V_\Gamma= \left( {\begin{array}{*{20}c}
   d & 0  \\
   0 & { - JdJ^*}  \\

 \end{array} } \right).
\]
Using $\V_\Gamma^{-1}=\S\V_\Gamma^{*}\S$ we then conclude that
\[
\S\V_\Gamma^{*}\S \A \S\V_\Gamma\S= \left( {\begin{array}{*{20}c}
   d & 0  \\
   0 & { JdJ^*}  \\

 \end{array} } \right)
\]
where $d\ge 0$ since $\A\ge 0$. Thus the unitarily implementable Bogoliubov transformation $\V_\A:=  \S\V_\Gamma \S$ diagonalizes $\A$. 

3. Finally we prove that $\A\Gamma = -\A\S 1_{(-\infty,0)}[\A \S]$. Denote 
\[
\V_\A^{ - 1} \Gamma (\V_\A^{ - 1} )^*  = \left( {\begin{array}{*{20}c}
   {\widetilde\gamma } & {\widetilde\alpha }  \\
   {(\widetilde\alpha )^* } & {J\widetilde\gamma J^* }  \\

 \end{array} } \right)~~\text{and}~ \V_\A^{ - 1} \Gamma' (\V_\A^{ - 1} )^*  = \left( {\begin{array}{*{20}c}
   {\widetilde\gamma ' } & {\widetilde\alpha '}  \\
   {(\widetilde\alpha ')^* } & {J\widetilde\gamma 'J^* }  \\

 \end{array} } \right)
\]
for any 1-dpm $\Gamma'$. We find that    
\bqq
  0 &\leqslant& \operatorname{Tr} \left[ {\A(\Gamma ' - \Gamma )} \right] = \operatorname{Tr} \left[ {(\V_\A^* \A\V_\A)\V_\A^{ - 1} (\Gamma ' - \Gamma )(\V_\A^* )^{ - 1} } \right] \hfill\\
   &=&   \operatorname{Tr} \left[ {\left( {\begin{array}{*{20}c}
   d & 0  \\
   0 & {JdJ^* }  \\

 \end{array} } \right)\left( {\begin{array}{*{20}c}
   {\widetilde\gamma'-\widetilde\gamma } & {\widetilde\alpha'-\widetilde\alpha }  \\
   {(\widetilde\alpha'-\widetilde\alpha )^* } & {J(\widetilde\gamma'-\widetilde\gamma) J^* }  \\

 \end{array} } \right)} \right] \hfill \\
  &=&   2\operatorname{Tr} [d\widetilde\gamma']-2\operatorname{Tr} [d\widetilde\gamma].
\eqq
Because this inequality holds true for any positive semi-definite trace class operator $\widetilde\gamma'$ on $\h$, we conclude that $\Tr(d\widetilde\gamma)=0$. By writing $\Tr[d\widetilde\gamma]=\Tr [(d^{1/2}\widetilde\gamma^{1/2})^* d^{1/2}\widetilde\gamma^{1/2}]$ we obtain $d^{1/2}\widetilde\gamma^{1/2}=0$, and hence $d\widetilde\gamma=0$. This also implies that $d\widetilde\alpha=0$ since 
$$(\widetilde\alpha d)^* (\widetilde\alpha d)=d(\widetilde\alpha)^*\widetilde\alpha d \le d (1+||\widetilde\gamma||_{\mathcal{L}(\h)}) \widetilde\gamma d=0.$$
Thus
\bqq
  \left( {\begin{array}{*{20}c}
   0 & 0  \\
   0 & {JdJ^* }  \\

 \end{array} } \right) &=& \left( {\begin{array}{*{20}c}
   d & 0  \\
   0 & {JdJ^* }  \\

 \end{array} } \right)\left( {\begin{array}{*{20}c}
   {\widetilde\gamma } & {\widetilde\alpha }  \\
   {(\widetilde\alpha )^* } & {1 + J\widetilde\gamma J}  \\

 \end{array} } \right) \hfill \\
   &=& (\V_\A^* \A\V_\A)(\V_\A^{ - 1} \Gamma (\V_\A^* )^{ - 1} ) = \V_\A^* \A\Gamma (\V_\A^* )^{ - 1} .
\eqq
It can be rewritten as 
\bq \label{eq:AGammaH}
\A\Gamma  = (\V_\A^* )^{ - 1} \left( {\begin{array}{*{20}c}
   0 & 0  \\
   0 & {JdJ^* }  \\

 \end{array} } \right)\V_\A^* .
\eq

Moreover, since $d\ge 0$ we find that 
\bqq
  \left( {\begin{array}{*{20}c}
   0 & 0  \\
   0 & {JdJ^* }  \\

 \end{array} } \right) &=&  - \left( {\begin{array}{*{20}c}
   d & 0  \\
   0 & {JdJ^* }  \\

 \end{array} } \right)1_{( - \infty ,0)} \left[ {\left( {\begin{array}{*{20}c}
   d & 0  \\
   0 & { - JdJ^* }  \\

 \end{array} } \right)} \right]S \hfill \\
   &=&  - (\V_\A^* \A\V_\A)1_{( - \infty ,0)} \left[ {\S\V_\A^* \A\V_\A} \right]\S \hfill\\
&=&  - \V_\A^* \A\V_\A 1_{( - \infty ,0)} \left[ {(\S\V_\A)^{ - 1} (\A\S)\S\V_\A} \right]\S \hfill \\
   &=&  - \V_\A^* \A\V_\A(\S\V_\A)^{ - 1} 1_{( - \infty ,0)} (\A\S)\S\V_\A\S \hfill\\
&=&  - \V_\A^* \A\S 1_{( - \infty ,0)} (\A\S)(\V_\A^* )^{ - 1} .
\eqq
Thus (\ref{eq:AGammaH}) can be rewritten as 
$$\A\Gamma=- \S\A 1_{( - \infty ,0)} [\A\S].$$

Moreover, it follows from (\ref{eq:AGammaH})  that
\[
\A\Gamma \S = (\V_\A^* )^{ - 1} \left( {\begin{array}{*{20}c}
   0 & 0  \\
   0 & {JdJ^* }  \\

 \end{array} } \right)\V_\A^* \S = \S\V_\A\left( {\begin{array}{*{20}c}
   0 & 0  \\
   0 & { - JdJ^* }  \\

 \end{array} } \right)\V_\A^* \S \leqslant 0. 
\]
\end{proof}

\subsection{Bogoliubov variational theory}

The Bogoliubov variational states should include not only the quasi-free states (like the Hartree-Fock theory) but also the coherent states, which correspond to the condensation. To describe the formulation precisely we need the following result (see \cite{So07}, Theorem 13.1).  
\begin{theorem}\label{thm:coherent-states} For every $\phi\in \h$ there exists (uniquely up to a complex phase)  a {\rm coherent unitary} $\U_\phi:\F\to \F$ such that 
$$\U_\phi^* a(f)\U_\phi=a(f)+(f,\phi)~~\text{for all}~f\in \h.$$
\end{theorem}

\begin{proof} We can proceed similarly the proof of Theorem \ref{thm:unitary-implementation} (see the Appendix) by translating the orthonormal basis $\left| {n_{i_1 } ,...,n_{i_M } } \right\rangle  = \left( {n_{i_1 } !...n_{i_M } !} \right)^{ - 1/2} a ^* (u_{i_M } )^{n_{i_M } } ...a^* (u_{i_1 } )^{n_{i_1 } } \left| 0 \right\rangle$ to
\[
\U_\phi \left| {n_{i_1 } ,...,n_{i_M } } \right\rangle = \left( {n_{i_1 } !...n_{i_M } !} \right)^{ - 1/2} [a^* (u_{i_M } )+(\phi,u_{i_M})]^{n_{i_M } } ...[a^* (u_{i_1 } )+(\phi,u_{i_1})]^{n_{i_1 } } \U_\phi\left| {0 } \right\rangle
\]
with the new vacuum 
\bqq
\U_\phi\left| {0 } \right\rangle  = \exp \left[ { -\frac{1}{2} \left\| \phi  \right\|^2 } \right]\exp \left[ { - a^* (\phi ) } \right]\left| 0 \right\rangle.
\eqq
\end{proof} 

\begin{remark} \begin{itemize}

\item[(i)] The condensate vector $\phi\in \h$ needs not be normalized. Any pure state $\left| \Psi  \right\rangle \left\langle \Psi  \right|$ with $\Psi=\U_\phi \left| 0 \right\rangle \in \F$ for some $\phi\in \h$ is called  a {\it coherent state}.

\item[(ii)] For generalized annihilation operators we get
$$ \U_\phi^* A(F)\U_\phi=A(F)+(F,\phi \oplus J\phi)_{\h\oplus \h^*} ~~\text{for all}~F\in \h\oplus \h^*.$$
\end{itemize}
\end{remark}

Now we can describe the Bogoliubov variational states. Denote

\[
{\G}^{\rm B} : = \left\{ {(\gamma ,\alpha) | \Gamma_{\gamma,\alpha} =\left(  {\begin{array}{*{20}c}
   \gamma  & \alpha   \\
   {J\alpha J} & {1 + J\gamma J^* }  \\

 \end{array} } \right) \geqslant 0,\operatorname{Tr} (\gamma ) < \infty} \right\}.
\]
The {\it Bogoliubov variational state} $\rho_{\gamma,\alpha,\phi}$ associated with $(\gamma,\alpha,\phi)\in \G^{Bo}\times \h $ is defined by 
$$
\rho_{\gamma,\alpha,\phi}(B):=\rho_{\gamma,\alpha}(\U_\phi^* B \U_\phi) ~~\text{for all}~B\in \B(\F),
$$
where $\rho_{\gamma,\alpha}$ is the quasi-free state with the 1-pdm $\Gamma_{\gamma,\alpha}$. In particular, the particle number expectation of the Bogoliubov variational state $\rho _{\gamma,\alpha ,\phi }$ is 
$$\rho _{\gamma,\alpha ,\phi } (\N) = \Tr(\gamma)+  ||\phi ||^2.$$

For a given Hamiltonian $\mathbb{H}:\F\to \F$ and $\lambda\ge 0$ we can define the {\it Bogoliubov ground state energy}
\bqq \label{eq:Bogoliubov-variational-problem}
E^{\rm B}(\lambda ) = \inf \left\{ {\rho_{\gamma,\alpha,\rho}(\mathbb{H})|(\gamma,\alpha,\phi) \in \G^{\rm B}\times \h,  \Tr(\gamma)+  ||\phi ||^2=\lambda} \right\}
\eqq
where $\rho_{\gamma,\alpha,\rho}(\mathbb{H})$ is the {\it Bogoliubov energy functional}  and $\lambda$ stands for the total particle number of the system.

\begin{remark} \begin{itemize}

\item[(i)] Due to the variational principle, the Bogoliubov ground state energy $E^{\rm B}_{\mathbb{H}}(\lambda)$ is always an upper bound to the quantum grand canonical energy 
$$
E^{\rm g}(\lambda)=\inf \{(\Psi, \mathbb{H} \Psi) : \Psi\in \F, ||\Psi || = 1, (\Psi, \N \Psi)=\lambda \}.
$$

\item[(ii)] If $N\in \mathbb{N}$ then the grand canonical energy $E^{\rm g}(N)$ is always a lower bound to the canonical energy 
$$E(N)=\inf \{(\Psi, \mathbb{H} \Psi) : \Psi\in \bigotimes _{\text{sym}}^N, ||\Psi || = 1\}.$$
Moreover, if $E(N)$ is convex w.r.t. $N$ then $E^{\rm g}(N)=E (N)$ for all $N$. 
\end{itemize}
\end{remark}

\begin{example}[A toy model] Let $\h=\R$ and the Hamiltonian 
$$\mathbb{H}=a^*(1)a^*(1)a(1)a(1).$$
A straightforward computation shows that for $N\in \mathbb{N}$ then the quantum energy is
$$ E^{\rm g}(N)=E(N)=N^2-N$$
and  the Bogoliubov energy is
\bqq
  {E^B}(N) &=& \mathop {\inf }\limits_{\lambda  \geqslant 0,x \geqslant 0,\lambda  + x = N} \left[ {{x^2} + x(4\lambda  - 2\sqrt {\lambda (1 + \lambda )} ) + 2{\lambda ^2} + \lambda (1 + \lambda )} \right] \hfill \\
   &=& \mathop {\inf }\limits_{0 \leqslant \lambda  \leqslant N} \left[ {{N^2} + 2N(\lambda  - \sqrt {\lambda (1 + \lambda )} ) + \lambda  + 2\lambda \sqrt {\lambda (1 + \lambda )} } \right] \hfill \\
   &=& {N^2} - N + O({N^{2/3}})~~~~\text{as}~N\to \infty.
   \eqq

\end{example}

Of our particular interest is the Bogoliubov variational theory for interacting Bose gases which we shall describe briefly below.

Let $\h=L^2(\Omega)$ for some measure space $\Omega$ with the inner product
\[
(u,v) = \int_\Omega {\overline {u(x)} v(x)dx}.
\]
In this case the mapping $J:\h\to \h^*$ is simply the complex conjugate, i.e. $Ju(x)=\overline{u(x)}$. Therefore, for simplicity we shall use notation $\overline{\gamma}=J \gamma J^*$ and $\overline{\alpha}=J\alpha J$. 

The Hamiltonian consists of a {\it one-body kinetic operator} $T$, which is a self-adjoint operator on $\h$, and a {\it two-body potential operator} $W$ which is the multiplication operator corresponding to the funtion $W(x,y):\Omega\times \Omega\to \R$ satisfying $W(x,y)=W(y,x)$. The grand canonical Hamiltonian ${\mathbb{H}}:\F\to \F$ can be represented in the second quantization as
\bqq
{\mathbb{H}}&=&\bigoplus_{N=0}^\infty \left( {\sum\limits_{i = 1}^N {T_i }  + \sum\limits_{1 \leqslant i < j \leqslant N} {W_{ij} } } \right) \\
&=& \sum\limits_{m,n} { (u_m ,Tu_n )_\h a_m^ *  a_n }  + \frac{1}
{2}\sum\limits_{m,n,p,q} {(u_m  \otimes u_n ,Wu_p  \otimes u_q )_{\h\times \h} a_m^ *  a_n^ *  a_p a_q }
\eqq
where $a_n:=a(u_n)$ and $\{u_n\}_{n= 1}^{\infty}$ is an orthonormal basis for $\h$ (the sum is independent of the choice of $\{u_n\}$). 

To represent the Bogoliubov energy functional explicitly in terms of $(\gamma,\alpha,\phi)$, it is convenient to introduce the {\it integral kernel} $\alpha(x,y)$ of the Hilbert-Schmidt operator which satisfies
\[
(\alpha f)(x) = \int_\Omega {\alpha (x,y)f(y)dy} ~~\text{for all}~ f\in L^2(\Omega). 
\]
Similarly, we have the kernel $\gamma(x,y)$ of $\gamma$ and the density functional is {\it formally} defined by $\rho_{\gamma}(x):=\gamma(x,x)$. More precisely, because $\gamma$ is a positive semi-definite trace class operator, we have the spectral decomposition $\gamma= \sum_i t_i \left| {u_i} \right\rangle \left\langle {u_i} \right|$ and then we can define $\gamma(x,y):= \sum_i t_i u_i(x) \overline{u_i(y)}$ and $\rho_\gamma(x):=\sum_i t_i |u_i(x)|^2.$ Note that $\int \rho(x)dx=\Tr(\gamma)$.

Using the coherent transformations
$$\U_\phi^* a_n U_\phi = a_n+(u_n,\phi)$$
and Wick's Theorem we find that
\bqq
\rho_{\gamma,\alpha,\phi}(\mathbb{H})&=&\operatorname{Tr} (T\widetilde \gamma ) + D({\rho _{\widetilde \gamma }},{\rho _{\widetilde \gamma }}) + X(\gamma ,\gamma ) + X(\alpha ,\alpha ) \hfill \\
   &~&+ \operatorname{Re} \iint\limits_{\Omega \times \Omega }  {\left[ {\gamma (x,y)\overline {\phi (x)} \phi (y) + \alpha (x,y)\overline {\phi (x)} \overline {\phi (y)} } \right]W(x,y)dxdy} \hfill \\ 
\eqq
where $\widetilde\gamma  := \gamma  + \left| \phi  \right\rangle \left\langle \phi  \right|$, $\rho_{\widetilde\gamma}(x)=\widetilde\gamma(x,x)$ and 
\[
D(f,g) = \frac{1}
{2}\iint\limits_{\Omega \times \Omega } {\overline{f(x)}g(y) W(x,y)dxdy},~~X(\gamma,\gamma' ) = \frac{1}
{2}\iint\limits_{\Omega \times \Omega } {\overline{\gamma (x,y)}\gamma' (x,y) W(x,y)dxdy}.
\]

Here are some specific examples with respect to three cases: $W>0$, $W$ changes sign, and $W<0$.

\begin{example}[Bosonic atoms] In this case we have
$$\h=L^2(\R^3),~ T=-\Delta-\frac{Z}{|x|},~W(x,y)=\frac{1}{|x-y|}.$$
We shall investigate the Bogoliubov theory for bosonic atoms in details in the next sections. In particular, we can show that the Bogoliubov ground state energy and the full quantum mechanics energy agree up to the leading order, and we conjecture that they even agree up to the second order.  
\end{example}

\begin{example}[Two-component Bose gases] This is the case when
$$\h=L^2(\R^3\times \{\pm 1\}),~ T=-\Delta_x,~W(x,e,y,e')=\frac{ee'}{|x-y|}.$$
It is already known that the Bogoliubov theory is also correct to the full quantum theory up to the leading order. More precisely, for large $N$, the correct leading term $-AN^{7/5}$ was predicted by Dyson \cite{Dy67} using the Bogoliubov principle and then it was mathematically established by Lieb-Solovej \cite{LS04} (lower bound) and Solovej \cite{So06} (upper bound).
\end{example}

\begin{example}[Bosonic stars] The system now corresponds to 
$$\h=L^2(\R^3),~ T=\sqrt{-\Delta+m^2}-m,~W(x,y)=-\frac{\kappa}{|x-y|}$$
where $m>0$ is the neutron mass and $\kappa=Gm^2>0$. Up to the leading order, the ground state energy is approximated by the Hartree model \cite{LY87}. Because the Hartree ground state energy is strictly concave, replacing the canonical setting by the grand canonical setting would make the energy much lower. Therefore, it is easy to see that the Bogoliubov ground state energy is much lower than the one of the full quantum model, although by adapting the ideas in \cite{LL10} we can show that the Bogoliubov variational model still has minimizers.  
\end{example}

\section{Bosonic atoms}

\subsection{Introduction}
For a bosonic atom we mean a system including a nucleus fixed at the origin in $\R^3$ with nucleus charge $Z>0$ and  N ``bosonic electrons" with charge $-1$. The system is described by the Hamiltonian 
\[
H_{N,Z} = \sum\limits_{i = 1}^N {\left( { - \Delta _i  - \frac{Z}
{{|x|}}} \right)}  + \sum\limits_{1 \leqslant i < j \leqslant N} {\frac{1}
{{|x_i  - x_j |}}} 
\]
acting on the symmetric space $\H_N=\bigotimes_{\text{sym}}^N L^2 
(\mathbb{R}^3 )$. The ground state energy of the system is given 
by 
\[
E(N,Z) = \inf \{ (\Psi ,H_{N,Z} \Psi )|\Psi  \in \H_N,||\Psi ||_{L^2}  = 1\}.
\]
In fact, the ground state energy   $E(N,Z)$ does not change if we replace the symmetric subspace $\H_N$ by the full $N$-particle space $\bigotimes^N L^2 (\mathbb{R}^3 )= L^{2}(\R^{3N})$ (see, e.g., \cite{LS10} p. 59-60). For usual atoms (with fermionic electrons), the Hamiltonian $H_{N,Z}$ acts on the anti-symmetric subspace $\mathop  \bigwedge \limits_{i = 1}^N (L^2 (\mathbb{R}^3 ) \otimes \mathbb{C}^2)$ instead. For simplicity, we only consider the spinless electrons because the spin number play no role in the mathematical analysis here. 

We recall some well-known fact about the full quantum problem. Due to the HVZ Theorem (see e.g. \cite{LS10} Lemma 12.1), $E(N,Z)\le E(N-1,Z)$ and if $E(N,Z)<E(N-1,Z)$ then $E(N,Z)$ is an isolated eigenvalue of $H_{N,Z}$. Unlike the asymptotic neutrality of fermionic atoms, in the bosonic case, the binding $E(N,Z)<E(N-1,Z)$ holds for all $0\le N \le N_c(Z)$ with $\lim_{Z\to \infty}N_c(Z)/Z=t_{\rm c}\approx 1.21$ (see \cite{BL83,So90,Ba84,Ba91}). 

The leading term of the ground state energy $E(N,Z)$ is given by the Hartree theory \cite{BL83}. In the Hartree theory, the ground state energy is 
\[
E^{\rm H} (N,Z) = \inf \{ \E^{\rm H} (u,Z):||u||_{L^2 }^2  = N\} 
\]
where
\[
\E^{\rm H} (u,Z) = \int\limits_{\mathbb{R}^3 } {\left| {\nabla u(x)} \right|^2 dx}  - \int\limits_{\mathbb{R}^3 } {\frac{{Z|u(x)|^2 }}
{{|x|}}}  + \iint\limits_{\mathbb{R}^3  \times \mathbb{R}^3 } {\frac{{|u(x)|^2 |u(y)|^2 }}
{{|x - y|}}dxdy}.
\]
By the scaling $u(x)=Z^2 u_1 (Zx)$ we have
$$\E^{\rm H} (u,Z)=Z^3 \E^{\rm H}(u_1,1).$$
Therefore, 
$$E^{\rm H}(N,Z)=Z^3 e(N/Z,1)~~\text{where}~e(t)= E^{\rm H} (t,1).$$
It is well-known (see e.g. \cite{Ba84, Li81}) that $e(t)$ is convex, $e(t)'<0$ when $t<t_c\approx 1.21$ and $e'(t)=0$ when $t\ge t_c$. Moreover, for any  $0<t<t_c\approx 1.21$, $e(t)$ has a unique minimizer $\phi_t$, which is positive, radially-symmetric and it is the unique solution to the nonlinear equation $h_t \phi_t=0$ where 
\[
h_t  =  - \Delta  - \frac{1}
{{|x|}} + |\phi _t^{} |^2 *\frac{1}
{{|x|}} - e'(t).
\]
As a consequence, $h_t\ge 0$. Moreover, since $\sigma_{\rm ess}(h_t)=[-e'(t),0]$, there is a gap $\Delta_t>0$ if $t<t_c$ such that $(h_t-\Delta_t)P_t^{\bot}\ge 0$ where $P_{\phi_t}^{\bot}=1-P_{t}$ with $P_{t}$ being the one-dimensional projection onto ${\rm Span}\{\phi_t\}$.

By scaling back, we conclude that $E^{\rm H}(tZ,Z)$ has the unique minimizer and the operator
\[
h_{t,Z}  =  - \Delta  - \frac{Z}
{{|x|}} + |\phi _{t,Z}|^2 *\frac{1}
{{|x|}} - Z^2e'(t)
\]
satisfies $h_{t,Z}\phi_{t,Z}=0$ and $(h_{t,Z}-Z^2\Delta_{t})P_{\phi_{t,Z}}^{\bot}\ge 0$ when $t<t_c$.

Our aim is to investigate the first correction to the ground state energy $E(tZ,Z)$. We shall analyze the Bogoliubov variational model for bosonic atoms and compare to the full quantum theory. From the general discussion on the Bogoliubov theory, we have the Bogoliubov variational problem
\bq \label{eq:Bogoliubov-variational-problem-bosonic-atom}
E^{\rm B}(N,Z) &=& \inf \left\{ {\E^{\rm B} (\gamma ,\alpha ,\phi,Z ) |(\gamma,\alpha,\phi) \in \G^{\rm B} \times L^2(\R^3),  \Tr(\gamma)+  ||\phi ||^2=N} \right\}
\eq
where
\bqq
\E^{\rm B}(\gamma ,\alpha ,\phi, Z ) &=& \operatorname{Tr} ( - [ \Delta-Z|x|^{-1}] \widetilde\gamma ) + D(\rho _{\widetilde\gamma}  ,\rho _{\widetilde\gamma}  ) +X(\gamma,\gamma) +X(\alpha,\alpha ) \hfill\\
&~&+\iint\limits_{{\mathbb{R}^3} \times {\mathbb{R}^3}} {\frac{{\gamma (x,y)\overline {\phi (x)} \phi (y)}}{{|x - y|}}dxdy} + \operatorname{Re} \iint\limits_{{\mathbb{R}^3} \times {\mathbb{R}^3}} {\frac{{\alpha (x,y)\overline {\phi (x)} \overline {\phi (y)} }}{{|x - y|}}dxdy}.
\eqq
Here we are using the notations $\widetilde\gamma  := \gamma  + \left| \phi  \right\rangle \left\langle \phi  \right|$ and 
\[
D(f,g) = \frac{1}
{2}\iint\limits_{\mathbb{R}^3  \times \mathbb{R}^3 } {\frac{{\overline {f(x)} g(y)}}
{{|x - y|}}dxdy},~~X(\gamma,\gamma' ) = \frac{1}
{2}\iint\limits_{\mathbb{R}^3  \times \mathbb{R}^3 } {\frac{{\overline{\gamma (x,y)}\gamma'(x,y) }}
{{|x - y|}}dxdy}.
\]

The properties of the Bogoliubov theory for bosonic atoms are the following, which will be proved in the next subsections.

\begin{theorem}[Existence of minimizers]\label{thm:existence-bosonic-atom} Let the nucleus charge $Z$ and the electron number $N$ be any positive numbers (not necessarily integers).  

\begin{itemize}

\item[(i)] If the binding inequality 
$$E^{\rm B}(N,Z)<E^{\rm B}(N',Z)~~\text{for all}~0<N'<N$$
holds then $E^B(N,Z)$ has a minimizer. 

\item[(ii)] The energy $E^{\rm B}(N,Z)$ is strictly decreasing on $N\in [0,N_c(Z)]$ with $N_c(Z)\ge Z$ for all $Z$ and 
$$\liminf_{Z\to \infty} \frac{N_c(Z)}{Z}\ge t_c\approx 1.21.$$
\end{itemize}
\end{theorem}

\begin{theorem}[Bogoliubov ground state energy]\label{thm:GSE-Bogoluibov-bosonic-atom} If $Z\to \infty$ and $N/Z=t \in (0,t_c)$ then 
\[
E^{\rm B} (N,Z) = Z^3 e(t) + Z^2 \mu (t) + o(Z^2 )
\]
where
\bqq 
\mu (t) := \mathop {\inf }\limits_{(\gamma ,\alpha ) \in \G^{\rm B} } \left[ { \Tr[h_t \gamma ] + \operatorname{Re} \iint\limits_{\mathbb{R}^3  \times \mathbb{R}^3 } {\frac{{[\gamma (x,y) + \alpha (x,y)]\phi _t (x)\phi _t (y)}}
{{|x - y|}}dxdy }} \right] .
\eqq
The coefficient $\mu(t)$ is finite and satisfies the lower bound
\[\mu (t) \leqslant {t^{ - 1}}e(t) - e'(t) + \widetilde \mu (t)\]
where
\[
\widetilde \mu (t) := \mathop {\min }\limits_{(\gamma ',\alpha ') \in {G^B},\gamma '{\phi _t} = 0} \left\{ {\operatorname{Tr} [{h_t}\gamma '] +\operatorname{Re} \iint\limits_{{\mathbb{R}^3} \times {\mathbb{R}^3}} {\frac{{[\gamma '(x,y) +  \alpha '(x,y)]{\phi _t}(x){\phi _t}(y)}}{{|x - y|}}}} \right\} <0 .\]
\end{theorem}

\begin{remark} \begin{itemize}

\item[(i)] If we restrict the Hamiltonian $H_{N,Z}$ into the class of $N$-particle product functions $\Psi_u=u\otimes u \otimes ... \otimes u$ then by scaling $u(x)=(N-1)^{-1/2}Z^2u_0(Zx)$ we have
\bqq
 \inf_{||u||=1} \left\langle {{\Psi _u},{H_{N,Z}}\Psi_u } \right\rangle  &=& \frac{{N{Z^3}}}{{N - 1}} \inf_{||u||_0=(N-1)/Z}{\E^{\rm H}}({u_0},1) =\frac{{N{Z^3}}}{{N - 1}}e\left( {\frac{{N - 1}}{Z}} \right) \hfill \\
   &=& {Z^3}e(t) + {Z^2}[{t^{ - 1}}e(t) - e'(t)]+o(Z^2). 
   \eqq
Because $\mu(t)< t^{-1}e(t)-e'(t)$, the Bogoliubov ground state energy is strictly lower than the lowest energy of the product wave functions at the second oder. 
   
   \item[(ii)] We believe, but do not have a rigorous proof, that the identity $\mu (t) ={t^{ - 1}}e(t) - e'(t) + \widetilde \mu (t)$ holds and a minimizing sequence of $\mu(t)$ is given by
$$
  \gamma  = \lambda \left| {\frac{{{\phi _t}}}{{||{\phi _t}||}}} \right\rangle \left\langle {\frac{{{\phi _t}}}{{||{\phi _t}||}}} \right| + \gamma ',~~
  \alpha  =  - \sqrt {\lambda (1 + \lambda )} \left| {\frac{{{\phi _t}}}{{||{\phi _t}||}}} \right\rangle \left\langle {\frac{{{\phi _t}}}{{||{\phi _t}||}}} \right| + \alpha ' 
$$
with $\lambda\to \infty$, where $(\gamma ',\alpha ')$ is a minimizer for $\widetilde \mu (t)$. In fact, the upper bound $\mu (t) \ge {t^{ - 1}}e(t) - e'(t) + \widetilde \mu (t)$ follows from the heuristic discussion on comparison between Bogoliubov energy and quantum energy below.    
   \end{itemize}
\end{remark}

We {\it conjecture} that the Bogoliubov theory determines the first correction to the quantum energy $E(N,Z)$. 

\begin{conjecture}[First correction to the leading energy]\label{conj:bosonic-atom-energy} If $Z\to \infty$ and $N/Z=t \in (0,t_c)$ then    
\[
E(N,Z)=E^{\rm B} (N,Z)+o(Z^2) = Z^3 e(t) + Z^2 \mu (t) + o(Z^2 ).
\]
\end{conjecture}

A heuristic discussion supporting the conjecture is made in the last subsection of the article. While the picture is rather clear, some technical work is still needed to make the argument rigorous. 

\subsection{Existence of Bogoliubov minimizers}

To prove the first claim of Theorem \ref{thm:existence-bosonic-atom}, we shall follow the extending variational argument (see e.g. \cite{LL01}, Theorem 11.12). Before studying the variational problem $E^{\rm B}(N,Z)$ in  (\ref{eq:Bogoliubov-variational-problem-bosonic-atom}), we start by considering the extended problem with the constraint $\operatorname{Tr} (\gamma ) \le N$, namely
\bq \label{eq:extended-variational-problem}
E^{\rm B}(_\le N,Z) = \inf \left\{ {\E^{\rm B} (\gamma ,\alpha ,\phi ,Z) |(\gamma,\alpha,\phi) \in \G^{\rm B},  \Tr(\gamma)+  ||\phi ||^2\le N} \right\}.
\eq 

\begin{lemma}[Extended problem]\label{le:extended-problem} The ground state energy $E^{\rm B}( N,Z)$ is finite and decreasing on $N$. Moreover, the extended variational problem $E^{\rm B}(_\le N,Z)$ in (\ref{eq:extended-variational-problem})  always has a minimizer.
\end{lemma}

\begin{proof} 1. By simply ignoring the non-negative two-body interaction and using the hydrogen bound, we have
\bq \label{eq:bosonic-atom-Bogoliubov-hydrogen-bound}
\E^{\rm B}(\gamma,\alpha,\phi,Z)\ge \operatorname{Tr} \left[ {\left( { - \Delta  - \frac{Z}{{|x|}}} \right)\widetilde \gamma } \right] \ge \frac{1}{2} \operatorname{Tr} ( - \Delta \gamma ) - \frac{Z^2N}{2}>-\infty.
\eq

2. Next, we prove that $E^{\rm B}(N',Z)\ge E^{\rm B}( N,Z)$ for $N'<N$. For any trial state $(\gamma,\alpha,\phi)$ with $(\gamma,\alpha)\in \G^{\rm B}$ and $\Tr \gamma + ||\phi||^2=N'$, choose $g\in C_c^\infty(\R^3)$ such that $\Tr(\gamma)+ ||\phi||^2+ ||g||^2=N$ and consider 
$$\gamma_\eps=\gamma+\left| g_\eps  \right\rangle \left\langle g_\eps  \right|$$
where $g_\eps(x)=\eps^{3/2}g(\eps x)$. Then $(\gamma_\eps,\alpha)\in \G^{\rm B}$ and $\Tr(\gamma_\eps)+||\phi||^2=N$. Moreover, since $|\nabla g_\varepsilon| \to 0$ in $L^2(\R^3)$ and $|g_\eps|^2\to 0$ in $L^p(\R^3)$ for any $p>1$, a simple calculation shows that 
$$\E^{\rm B}(\gamma_\eps,\alpha,\phi,Z)\to \E^{\rm B}_Z(\gamma,\alpha,\phi,Z)~~\text{as}~\eps\to 0.$$
This ensures that $E^{\rm B}(N',Z)\ge E^{\rm B}(N',Z)$. 

3. To show that $E^{\rm B}(_\le N,Z)$ has a minimizer, let us take a minimizing sequence $(\gamma_n,\alpha_n,\phi_n)$ for $E^{\rm B}(_\le N,Z)$. The lower bound (\ref{eq:bosonic-atom-Bogoliubov-hydrogen-bound}) ensures that $\operatorname{Tr} ( - \Delta \gamma_n )$ is bounded. Consequently, all of $\left\| {\gamma _n (x,y)} \right\|_{H^1 (\mathbb{R}^3  \times \mathbb{R}^3 )}$, $\left\| {\alpha _n (x,y)} \right\|_{H^{1/2} (\mathbb{R}^3  \times \mathbb{R}^3 )}$ and $\left\| {\phi _n } \right\|_{H^1 (\mathbb{R}^3 )}$ are bounded. By passing to a subsequence if necessary, we may assume that $\gamma_n \wto \gamma,\alpha_n\wto \gamma,\phi_n\wto \phi$ weakly in the corresponding Hilbert spaces, and their kernels converge pointwisely. It is straightforward to check that $(\gamma,\alpha)\in \G^{\rm B}$ and by Fatou's lemma, $\Tr(\gamma)+||\phi||^2 \le N$. 

Fatou'lemma also implies that    
$$\mathop {\lim \inf }\limits_{n \to \infty } \operatorname{Tr} (-\Delta\gamma ) \ge \operatorname{Tr} (-\Delta\gamma ).$$
The two-body interaction part of $\E^{\rm B}(\gamma_n,\alpha_n,\phi_n,Z)$ can be rewritten as
\[\iint {\frac{{W({\gamma _n},{\alpha _n},{\phi _n})}}{{|x - y|}}}\]
where
\[\begin{gathered}
  W({\gamma _n},{\alpha _n},{\phi _n}) = {\rho _{{\gamma _n}}}(x){\rho _{{\gamma _n}}}(y) + |{\gamma _n}(x,y){|^2} + |{\alpha _n}(x,y) + {\phi _n}(x){\phi _n}(y){|^2} \hfill \\
   + \left[ {{\rho _{{\gamma _n}}}(x)|\phi (y){|^2} + {\rho _{{\gamma _n}}}(y)|\phi (x){|^2} + 2\operatorname{Re} ({\gamma _n}(x,y)\overline {{\phi _n}(x)} {\phi _n}(y))} \right] \geqslant 0. \hfill \\ 
\end{gathered} \]
Therefore, we may use Fatou'lemma again to obtain
\[
\lim \inf_{n\to \infty} \iint {\frac{{W({\gamma _n},{\alpha _n},{\phi _n})}}{{|x - y|}}} \geqslant \iint {\frac{{W(\gamma ,\alpha ,\phi )}}{{|x - y|}}}.\]
Finally, because ${\sqrt {\rho _{\gamma_n}  } }\wto \sqrt {\rho _{\gamma}  }$ in $H^1(\R^3)$ we have the convergence
\[
\int\limits_{\mathbb{R}^3 } {\frac{{\rho _{\gamma _n } (x)}}
{{|x|}}dx}  \to \int\limits_{\mathbb{R}^3 } {\frac{{\rho _\gamma  (x)}}
{{|x|}}dx} ~~\text{as}~n\to \infty.
\]
Therefore, we have 
\[
\mathop {\lim \inf }\limits_{n \to \infty } \E^{\rm B}(\gamma _n ,\alpha _n ,\phi _n,Z ) \geqslant \E^{\rm B}(\gamma ,\alpha ,\phi ,Z)
\]
and hence $(\gamma ,\alpha ,\phi )$ is a minimizer for $E^{\rm B}(_\le N,Z)$. 
\end{proof}

We now prove the existence of minimizers for the original problem $E^{\rm B}(N,Z)$. 

\begin{proof}[Proof of Theorem \ref{thm:existence-bosonic-atom}] 1. If $E^{\rm B}(N,Z)<E^{\rm B}(N',Z)$ for all $0<N'<N$ then any minimizer $(\gamma,\alpha,\phi)$ for the extended problem $E^{\rm B}(_\le N,Z)$ must satisfy $\Tr(\gamma)+||\phi||^2=N$, and hence it is a minimizer for $E^{\rm B}(N,Z)$.

2. That $E^{\rm}(N,Z)$ is strictly decreasing on $N\in [0,Z]$ follows by the same argument as in \cite{FLSS07}. Assume that $E^{\rm B}(N,Z)=E^{\rm B}(N',Z)$ for some $0\le N'<N\le Z$. Let $(\gamma,\alpha,\phi)$ be a minimizer for $E^{\rm B}( _\le N',Z)$. For any $\varphi\in H^1(\R^3)$, let us consider the trial state $(\gamma_\eps,\alpha,\phi)$ with
\[
\gamma _\varepsilon   = \gamma  + \varepsilon \left| \varphi  \right\rangle \left\langle \varphi  \right|,~~\eps>0.
\]
For $\eps>0$ small we have $\Tr \gamma_\eps +||\phi||\le N$ and hence
$$\E^{\rm B}(\gamma _\varepsilon,\alpha,\phi,Z)\ge E^{\rm B}(N,Z)=E^{\rm B}(N',Z)=\E^{\rm B}(\gamma,\alpha,\phi,Z).$$
Therefore,
\bq \label{eq:0<=diffE(gamma-e)}
0\le \left. {\frac{d}
{{d\varepsilon }}} \right|_{\varepsilon  = 0} \E^{\rm B}(\gamma _\varepsilon  ,\alpha ,\phi,Z ) &=& (\varphi , - \Delta \varphi )_{L^2 }  - \int\limits_{\mathbb{R}^3 } {\frac{{Z|\varphi (x)|^2 }}
{{|x|}}dx}  + 2D(\rho _{\widetilde\gamma}  ,|\varphi |^2 )\nn\hfill\\
&~& + 2\operatorname{Re} X(\widetilde\gamma ,\left| \varphi  \right\rangle \left\langle \varphi  \right|).
\eq

On the other hand, let us replace $\varphi$ by $\varphi_L(x):=L^{-3/2}\varphi_1(x/L)$ where $\varphi_1\in H^1(\R^3)$ such that $\varphi_1$ is radially-symmetric and $\varphi_1(x)=0$ if $|x|<1$ and $\varphi_1(x)>0$ if $|x|>1$. Then for large $L$ one has
\bqq
\left\langle {\varphi _L , - \Delta \varphi _L } \right\rangle  &=& L^{ - 2} \left\langle {\varphi_1 , - \Delta \varphi_1 } \right\rangle  = O(L^{ - 2} ),\hfill\\
 - Z\int\limits_{\mathbb{R}^3 } {\frac{{|\varphi _L (x)|^2 }}
{{|x|}}dx}  &=&  - ZL^{ - 1} \int\limits_{\mathbb{R}^3 } {\frac{{|\varphi_1 (x)|^2 }}
{{|x|}}dx}.
\eqq
Moreover, by Newton's theorem,
$$
2D(\rho _{\widetilde\gamma}  ,|\varphi_L|^2 )=\iint\limits_{\R^3\times \R^3 } {\frac{{\rho_\gamma(x)|\varphi_L (y)|^2 }}
{{\max \{ |x|,|y|\}}}dy}\leqslant N'L^{-1} \int\limits_{\mathbb{R}^3 } {\frac{{|\varphi_1 (y)|^2 }}
{{|y|}}dy},
$$
and by H\"older's inequality,
\[
\begin{gathered}
  2\operatorname{Re} X(\widetilde\gamma ,\left| {\varphi _L } \right\rangle \left\langle {\varphi _L } \right|) = \iint\limits_{\mathbb{R}^3  \times \mathbb{R}^3 } {\frac{{\overline {\widetilde\gamma (x,y)} \overline {\varphi _L (x)} \varphi _L (y)}}
{{|x - y|}}dxdy} \hfill \\
   \leqslant \left( {\iint\limits_{|x| \geqslant L,|y| \geqslant L} {|\widetilde \gamma (x,y)|dxdy}} \right)^{1/2} \left( {\iint\limits_{\mathbb{R}^3  \times \mathbb{R}^3 } {\frac{{|\varphi _L (x)|^2 |\varphi _L (y)|^2 }}
{{|x - y|^2 }}dxdy}} \right)^{1/2}  = o(L^{ - 1} ). \hfill \\ 
\end{gathered} 
\]
Thus if we replace $\varphi$ in (\ref{eq:0<=diffE(gamma-e)}) by $\varphi_L$ then we obtain 
\[
0 \leqslant O(L^{ - 2} ) - (Z - N')L^{ - 1} \int\limits_{\mathbb{R}^3 } {\frac{{|\varphi_1 (x)|^2 }}
{{|x|}}dx}  + o(L^{ - 1} )
\]
which is a contradiction to the assumption $N'<Z$. Thus $N \mapsto E^{\rm B}(N,Z)$ is strictly decreasing when $0<N\le Z$.

3. Now we show that $E^{\rm B}(Z,N)$ is strictly decreasing on $N\in [Z,N_c(Z)]$ with $$\liminf_{Z\to \infty} N_c(Z)/Z\ge t_c\approx 1.21$$
We shall need some properties of the Bogoliubov ground state in Lemma \ref{le:estimate-ground-state}, which is derived in the next section. 

Take a large number $Z$ and assume that $N\mapsto E^{\rm B}(N,Z)$ is not strictly decreasing on $Z\le t'Z$ for a fixed value $t'<t_c$. Then there exists $N=tZ \in [Z,t'Z]$ and $\delta>0$ such that $E^{\rm B}(N,Z)=E^{\rm B}(N+\delta,Z)$ and $E^{\rm B}(N,Z)$ has a ground state $(\gamma,\alpha,\phi)$. Because
$$ \E^{\rm B}(\gamma,\alpha,\phi,Z)=E^{\rm B}(N,Z)=E^{\rm B}(N+\delta,Z)\le \E^{\rm B}(\gamma,\alpha,\sqrt{1+\eps}\phi,Z)$$
for $\eps>0$ small, we have
\bqq
  0  &\le & \left. {\frac{d}
{{d\varepsilon }}} \right|_{\varepsilon  = 0} \E^{\rm B}(\gamma ,\alpha ,\sqrt {1 + \varepsilon } \phi,Z ) \hfill \\
   &=& \left\langle {\phi ,\left( { - \Delta  - \frac{Z}{{|x|}} + {\rho _{\widetilde \gamma }}*|.{|^{ - 1}}} \right)\phi } \right\rangle  + \iint {\frac{{\gamma (x,y)\overline {\phi (x)} \phi (y)}}{{|x - y|}}} + \operatorname{Re} \iint {\frac{{\alpha (x,y)\overline {\phi (x)} \overline {\phi (y)} }}{{|x - y|}}}.
\eqq
Because 
$$ \E^{\rm B}(\gamma,\alpha,\phi)=E^{\rm B}(N,Z) \le E^{\rm H}(N,Z)\le \operatorname{Tr} \left[ {\left( { - \Delta  - \frac{Z}{{|x|}}} \right)\widetilde \gamma } \right] + D({\rho _{\widetilde \gamma }},{\rho _{\widetilde \gamma }})$$
we get
\[\iint {\frac{{\gamma (x,y)\overline {\phi (x)} \phi (y)}}{{|x - y|}}} + \operatorname{Re} \iint {\frac{{\alpha (x,y)\overline {\phi (x)} \overline {\phi (y)} }}{{|x - y|}}} \leqslant 0.\]
Thus
\[0 \leqslant \left\langle {\phi ,\left( { - \Delta  - \frac{Z}{{|x|}} + {\rho _{\widetilde \gamma }}*|.{|^{ - 1}}} \right)\phi } \right\rangle  = \left\langle {\phi ,{h_{t,Z}}\phi } \right\rangle  + 2D({\rho _{\widetilde \gamma }} - |{\phi _{t,Z}}{|^2},|\phi {|^2}) + e'(t){Z^2}||\phi |{|^2}.\]

On the other hand, using the estimates in Lemma \ref{le:estimate-ground-state} we have
\bqq
  \left\langle {\phi ,{h_{t,Z}}\phi } \right\rangle  &=& o({Z^2}), \hfill\\
e'(t){Z^2}||\phi |{|^2} &\leqslant & e'(t){Z^2}(tZ + o(Z)) = te'(t){Z^3} + o({Z^3}), \hfill \\
  D({\rho _{\widetilde \gamma }} - |{\phi _{t,Z}}{|^2},|\phi {|^2}) &\leqslant & \sqrt {D({\rho _{\widetilde \gamma }} - |{\phi _{t,Z}}{|^2},{\rho _{\widetilde \gamma }} - |{\phi _{t,Z}}{|^2})} .\sqrt {D(|{\phi _{t,Z}}{|^2},|{\phi _{t,Z}}{|^2})}  = o({Z^{5/2}}). \hfill 
\eqq  
Therefore,
\[0 \leqslant \left\langle {\phi ,{h_{t,Z}}\phi } \right\rangle  + 2D({\rho _{\widetilde \gamma }} - |{\phi _{t,Z}}{|^2},|\phi {|^2}) + e'(t){Z^2}||\phi |{|^2} \leqslant te'(t){Z^3} + o({Z^3}).\]
However, it is a contradiction because $te'(t)<0$ when $1\le t\le t'<t_c$. 
\end{proof}

\subsection{Analysis of quadratic forms}

We consider the minimization problem $\mu(t)$ of the quadratic form in Theorem \ref{thm:GSE-Bogoluibov-bosonic-atom}. Recall that
\bqq 
\mu (t) := \mathop {\inf }\limits_{(\gamma ,\alpha ) \in \G^{\rm B} } q_t(\gamma,\alpha) ~~\text{and} ~\widetilde \mu(t) := \mathop {\inf }\limits_{(\gamma' ,\alpha' ) \in \G^{\rm B},\gamma ' \phi_t =0 } q_t(\gamma',\alpha') 
\eqq
where 
$$
q_t(\gamma,\alpha):=\left[ { \Tr[h_t \gamma ] + \operatorname{Re} \iint\limits_{\mathbb{R}^3  \times \mathbb{R}^3 } {\frac{{[\gamma (x,y) + \alpha (x,y)]\phi _t (x)\phi _t (y)}}
{{|x - y|}}dxdy }} \right] .
$$
\begin{lemma}[Analysis of the quadratic form $q_t(\gamma,\alpha)$]\label{le:mu-t-finite} For any $0<t<t_c$ we have 
$$-\infty<\mu(t)\le t^{-1}e(t)-e'(t)+\widetilde \mu(t).$$
Moreover, the minimization problem $\widetilde \mu(t)$ has a minimizer $(\gamma',\alpha')$ and $\widetilde \mu(t)<0$.
\end{lemma}
\begin{proof} 1. Because $q_t(\gamma,\alpha)$ is a quadratic form of $(\gamma,\alpha)$, for considering the ground state energy we may restrict $(\gamma,\alpha)$ into the class of quasi-free pure state, i.e. $\alpha\alpha^*=\gamma(1+\gamma)$. Since $\gamma\ge 0$ is trace class and $\alpha^T=\alpha$, we can write  
\[
\gamma (x,y)  = \sum\limits_n {\lambda _n {u_n(x)}\overline{u_n(y)}},~ \alpha (x,y)  =   -\sum\limits_n {\sqrt {\lambda _n (1 + \lambda _n )} u_n(x)u_n(y)} ,
\]
where $\lambda_n\ge 0$ and $\{u_n\}_n$ is an orthonormal family on $L^2(\R^3)$. Then
\bqq
  q_t (\gamma ,\alpha) =\sum\limits_n [\lambda _n (u_n ,h_t u_n ) + A_n]
\eqq
with
\bqq
  A_n  =  \lambda _n \iint\limits_{\mathbb{R}^3  \times \mathbb{R}^3 } {\frac{{u_n (x)\overline {u_n (y)}  {\phi_t(x)} \phi_t (y)}}
{{|x - y|}}} -\sqrt {\lambda _n (1 + \lambda _n )} \operatorname{Re} \iint\limits_{\mathbb{R}^3  \times \mathbb{R}^3 } {\frac{{u_n (x)u_n (y) {\phi_t (x)\phi_t (y)} }}
{{|x - y|}}}.
\eqq

2. We may assume that $\lambda_n (u_n,h_{t,Z}u_n)+A_n \le 0$ for all $n$; otherwise, if $\lambda_n (u_n,h_{t,Z}u_n)+A_n<0$ then 
$$q_{t}(\gamma,\alpha)> q_{t}(\gamma ',\alpha')$$
where
$$\gamma'=\gamma-\lambda_n \left| {u_n } \right\rangle \left\langle {u_n } \right| ,\alpha'=\alpha +\sqrt{\lambda_n(1+\lambda_n)} \left| {u_n } \right\rangle \left\langle {\overline{u_n} } \right|.
$$

We have the gap $(u_n,h_{t}u_n) \ge \Delta_t||P^{\bot}u_n||^2$ for all $n$. Moreover, $|{\rm Re}(D(\overline{u},u))|\le D(u,u)$ for all functions $u$ we have   
\bq \label{eq:estimate-mu-t-An}
A_n\ge 2D(u_n \phi _t ,u_n \phi _t )(\lambda _n  - \sqrt {\lambda _n (1 + \lambda _n )} )\ge 2D(u_n \phi _t ,u_n \phi _t ) \max \left\{ { - \frac{1}
{2}, - \sqrt {\lambda _n } } \right\}.
\eq
Thus it follows from the assumption $\lambda_n (u_n,h_{t,Z}u_n)+A_n \le 0$ that
\bq \label{eq:estimate-mu-t-h-An-1}
\Delta_t^2 \lambda _n ||P^ \bot  u_n ||^4  \leqslant 4|D(u_n \phi _t ,u_n \phi _t )|^2~~\text{for all}~n.
\eq

On the other hand, observe that  
$$||P^{\bot}u_n||^2+||P^{\bot}u_m||^2 =2-||Pu_n||^2-||Pu_m||^2 \ge 1~~\text{for all}~m\ne n.$$
Therefore, there exists (at most) an element $i_0$ such that $||P^{\bot}u_n||^2\ge 1/2$ for all $n\ne i_0$. As a consequence, (\ref{eq:estimate-mu-t-h-An-1}) implies that
$$ \sum_{n\ne i_0} \lambda_n \le 16\Delta_t^{-2}\sum_{n\ne i_0}|D(u_n \phi _t ,u_n \phi _t )|^2\le 4\Delta_t^{-2} \iint\limits_{\mathbb{R}^3  \times \mathbb{R}^3 } {\frac{{|\phi _t (x)|^2 |\phi _t (y)|^2 }}
{{|x - y|^2 }}dxdy} \le C.$$

3. Using $h_t\ge 0$ and (\ref{eq:estimate-mu-t-An}) we have
\bqq q_t(\gamma,\alpha) &\ge & A_{i_0}+\sum_{n\ne i_0}A_n 
\ge -D(u_{i_0}\phi_t,u_{i_0}\phi_t) -\sum_{n\ne i_0}2\sqrt{\lambda_n}D(u_{n}\phi_t,u_{n}\phi_t)\hfill\\
&\ge & -D(u_{i_0}\phi_t,u_{i_0}\phi_t) -2\left( {\sum_{n\ne i_0} \lambda_n} \right)^{1/2}\left( {\sum_{n\ne i_0} |D(u_n\phi_t,u_n\phi_t) |^2} \right)^{1/2}\hfill\\
&\ge & -\frac{1}{2} \left( { \iint\limits_{\mathbb{R}^3  \times \mathbb{R}^3 } {\frac{{|\phi _t (x)|^2 |\phi _t (y)|^2 }}
{{|x - y|^2 }}} } \right)^{1/2}- 2\Delta_t^{-1} \left( { \iint\limits_{\mathbb{R}^3  \times \mathbb{R}^3 } {\frac{{|\phi _t (x)|^2 |\phi _t (y)|^2 }}
{{|x - y|^2 }}} } \right) \ge -C.
\eqq

4. To see the upper bound on $\mu(t)$ let us consider the trial state 
$$
  \gamma  = \lambda \left| {\frac{{{\phi _t}}}{{||{\phi _t}||}}} \right\rangle \left\langle {\frac{{{\phi _t}}}{{||{\phi _t}||}}} \right| + \gamma ',~~
  \alpha  =  - \sqrt {\lambda (1 + \lambda )} \left| {\frac{{{\phi _t}}}{{||{\phi _t}||}}} \right\rangle \left\langle {\frac{{{\phi _t}}}{{||{\phi _t}||}}} \right| + \alpha ' 
$$
 where $(\gamma',\alpha')\in \G^{\rm B}$ such that $\gamma' \phi_t=0$. One has
\bqq
  \mu (t) &\leqslant & {q_t}(\gamma ,\alpha ) = 2\left( {{\lambda } - \sqrt {{\lambda }(1 + {\lambda})} } \right) D({u_1}{\phi _t},{u_1}{\phi _t})+q_t(\gamma',\alpha').
\eqq
Taking the infimum over all $(\gamma',\alpha')$ and letting  $\lambda\to \infty$ we obtain
$$
\mu (t) \leqslant  - t^{-1}D(|\phi _t|^2, |\phi _t|^2)+\widetilde \mu (t) = {t^{ - 1}}e(t) - e'(t)+\widetilde \mu (t).
$$ 

5. Now we consider $\widetilde \mu (t)$. The above argument shows that if $\{(\gamma_n',\alpha_n')\}_{n=1}^\infty$ is a minimizing sequence for $\widetilde \mu (t)$ then $\Tr(\gamma_n')$ is bounded. Therefore, it follows from the standard compactness argument that $\widetilde \mu (t)$ has a minimizer. To see that $\widetilde \mu (t)<0$, let us consider
\[\gamma ' = \lambda '\left| u \right\rangle \left\langle u \right|,~~\alpha ' =  - \sqrt {\lambda '(1 + \lambda ')} \left| u \right\rangle \left\langle u \right|\]
where $u$ is a normalized real-valued function in $L^2(\R^3)$ such that $(u,\phi_t)=0$. Because $D({u}{\phi _t},{u}{\phi _t})>0$ we have 
$$ \widetilde \mu (t)\le q_t(\gamma',\alpha')={\lambda'}({u},{h_t}{u}) + 2\left( {{\lambda'} - \sqrt {{\lambda'}(1 + {\lambda '})} } \right)t^{-1}D({u}{\phi _t},{u}{\phi _t})<0$$
for some $\lambda'>0$ small enough.
\end{proof}

\begin{remark} The analysis here works out for a more general setting. For example, if $h$ is a positive semi-definite operator on $L^2(\Omega)$ with $\inf \sigma_{\rm ess}(h)>0$ and $W$ is a positive semi-definite Hilbert-Schmidt operator on $L^2(\Omega)$ with a real-valued kernel $W(x,y)$ then  
\[
\mathop {\inf } \limits_{(\gamma ,\alpha ) \in \G^{\rm B} }{\left( {\Tr[h \gamma ] + \operatorname{Re} \iint\limits_{\mathbb{R}^3  \times \mathbb{R}^3 } {{{(\gamma (x,y) + \alpha (x,y)) W(x,y)}}
dxdy}} \right)
} >-\infty .
\]
\end{remark}

By scaling $\phi_{t,Z}(x)=Z^2\phi_t(Zx)$, $\gamma(x)=Z^3\gamma'(Zx,Zy)$, $\alpha(x)=Z^3\alpha'(Zx,Zy)$ we have
\[
\mathop {\inf }\limits_{(\gamma ,\alpha ) \in \G^{\rm B} } q_{t,Z} (\gamma ,\alpha ) = \mathop {\inf }\limits_{(\gamma ',\alpha ') \in \G^{\rm B} } Z^2 q_{t} (\gamma' ,\alpha' )=Z^2 \mu (t)
\]
where
\[
q_{t,Z} (\gamma ',\alpha ') = \operatorname{Tr} [h_{t,Z} \gamma '] + \iint\limits_{\mathbb{R}^3  \times \mathbb{R}^3 } {\frac{{(\gamma '(x,y) + \alpha '(x,y))\phi _{t,Z} (x)\phi _{t,Z} (y)}}
{{|x - y|}}dxdy}.
\]

To prove Theorem \ref{thm:GSE-Bogoluibov-bosonic-atom}, we need to consider some perturbation form of $q_{t,Z}$. 

\begin{lemma}[Analysis of pertubative quadratic forms]\label{le:pert-quadratic} Let $\phi\in L^2(\R^3)$ such that $||\phi||\le ||\phi_{t,Z}||$, $||\nabla \phi||\le CZ^{3/2}$ and $||P^{\bot}\phi||\le C$ where $P^{\bot}=1-P$ with $P$ being the one-dimensional projection onto $\phi_{t,Z}$. Then for $Z$ large we have
\[
\mathop {\inf }\limits_{(\gamma ,\alpha ) \in \G^{\rm B} } q_{t,Z} (\gamma ,\alpha ,\phi ) \geqslant \frac{{\left\| {P\phi } \right\|^2 }}
{{||\phi _{t,Z} ||^2 }}  Z^2 \mu (t) - CZ^{2 - 1/10 } 
\]
where
\[
q_{t,Z} (\gamma ,\alpha ,\phi ) = \operatorname{Tr} [h_{t,Z} \gamma ] + \iint\limits_{\mathbb{R}^3  \times \mathbb{R}^3 } {\frac{{\gamma (x,y)\overline {\phi (x)} \phi (y)}}
{{|x - y|}}} + \iint\limits_{\mathbb{R}^3  \times \mathbb{R}^3 } {\frac{{\alpha (x,y)\overline {\phi (x)} \overline {\phi (y)} }}
{{|x - y|}}}.
\]
\end{lemma} 

\begin{proof} 1. We first consider the case when $\Tr \gamma$ is small. Assume that $\Tr \gamma \le Z^{1/2-\eps}$, where $\eps=1/10$. In the integral involved with $\gamma$, we use the decomposition
\[
\overline {\phi (x)} \phi (y) = \overline {P\phi (x)} P\phi (y) + \overline {P\phi (x)} P^ \bot  \phi (y) + \overline {P^ \bot  \phi (x)} \phi (y).
\]
Observe that all terms involved with $P^ \bot  \phi$ have negligible contribution. For example,
\[
\left| {\iint\limits_{\mathbb{R}^3  \times \mathbb{R}^3 } {\frac{{\gamma (x,y)\overline {P^ \bot  \phi (x)} \phi (y)}}
{{|x - y|}}}} \right| \leqslant 2\operatorname{Tr} (\gamma _{}^2 )^{1/2} ||P^ \bot  \phi ||.||\nabla \phi || \leqslant CZ^{2 - \varepsilon } .
\]
Thus
\[
\iint\limits_{\mathbb{R}^3  \times \mathbb{R}^3 } {\frac{{\gamma _\varepsilon  (x,y)\overline {\phi (x)} \phi (y)}}
{{|x - y|}}} \geqslant \iint\limits_{\mathbb{R}^3  \times \mathbb{R}^3 } {\frac{{\gamma _\varepsilon  (x,y)\overline {P\phi (x)} P\phi (y)}}
{{|x - y|}}} - CZ^{2 - \varepsilon }. 
\]

Together with the similar bound on the integral  involved with $\alpha$, we arrive at
\bqq
  q_{t,Z} (\gamma ,\alpha ,\phi ) &\geqslant & \left( {1 - \frac{{||P\phi ||^2 }}
{{||\phi ||^2 }}} \right)\operatorname{Tr} [h_{t,Z} \gamma ] + \frac{{||P\phi ||^2 }}
{{||\phi _{t,Z}||^2 }}q_{t,Z} (\gamma ,\alpha ) - CZ^{2 - \varepsilon }  \hfill \\
  & \geqslant & \frac{{||P\phi ||^2 }}
{{||\phi_{t,Z} ||^2 }}Z^2 \mu (t) - CZ^{2 - \varepsilon }  .
\eqq

2. Now we consider the case when $\Tr \gamma $ is large. Assume $\Tr \gamma \ge Z^{1/2-\eps}$. Following the proof of Lemma \ref{le:mu-t-finite}, we may assume that
\[
\gamma  = \lambda _1 \left| {u_1 } \right\rangle \left\langle {u_1 } \right| + \gamma ',\alpha  =  - \sqrt {\lambda _1 (1 + \lambda _1 )} \left| {u_1 } \right\rangle \left\langle {\overline {u_1 } } \right| + \alpha '
\]
where $||u_1||=1$ and $(\gamma',\alpha')$ is the 1-pdm of a pure quasi-free state such that
$$\Tr \gamma' \le C,~\lambda_1 ||P^{\bot}u_1||^2\le C ~~\text{and}~\gamma'u_1=0=\alpha'u_1.$$
Because $\lambda_1=\Tr \gamma-\Tr \gamma'\ge Z^{1/2-\eps}-C$ and $\lambda_1 ||P^{\bot}u_1||^2\le C$, we have 
$$||P^{\bot}u_1||^2\le CZ^{-1/2+\eps}.$$
As a consequence,
$$ \Tr(P\gamma')=\sum_{n\ne 1} \lambda_n ||Pu_n||^2 \le \Tr \gamma' \sum_{n\ne 1} ||Pu_n||^2 \le  \Tr \gamma' ||P^{\bot} u_1||^2 \le Z^{-1/2+\eps}.$$ 

3. We shall compare $q_{t,Z}(\gamma,\alpha,\phi)$ with $q_{t,Z}(\gamma'',\alpha'')$ where
\[
\begin{gathered}
  \gamma '' = {\lambda _1 P\left| {u_1 } \right\rangle \left\langle {u_1 } \right|P + P^ \bot  \gamma 'P^ \bot  }, \hfill \\
  \alpha '' =  - \sqrt {\lambda _1 (1 + \lambda _1 )} P\left| {u_1 } \right\rangle \left\langle {\overline {u_1 } } \right|P + P^ \bot  \alpha 'P^ \bot   .\hfill \\ 
\end{gathered} 
\]
It is easy to see that $(\gamma'',\alpha'')\in \G^{\rm B}$.

We first consider the terms involved with $u_1$. We have
\[
\begin{gathered}
  \lambda _1 \iint\limits_{\mathbb{R}^3  \times \mathbb{R}^3 } {\frac{{u_1 (x)\overline {u_1 (y)} \overline {\phi (x)} \phi (y)}}
{{|x - y|}}} - \sqrt {\lambda _1 (1 + \lambda _1 )} \operatorname{Re} \iint\limits_{\mathbb{R}^3  \times \mathbb{R}^3 } {\frac{{u_1 (x)u_1 (y)\overline {\phi (x)} \overline {\phi (y)} }}
{{|x - y|}}} \hfill \\
   \geqslant (\lambda _1  - \sqrt {\lambda _1 (1 + \lambda _1 )} )\iint\limits_{\mathbb{R}^3  \times \mathbb{R}^3 } {\frac{{u_1 (x)\overline {u_1 (y)} \overline {\phi (x)} \phi (y)}}
{{|x - y|}}}. \hfill \\ 
\end{gathered} 
\]
Then we use the decomposition
\bqq
  u_1 (x)\overline {u_1 (y)}  &=& Pu_1 (x)\overline {Pu_1 (y)}  + Pu_1 (x)\overline {P^ \bot  u_1 (y)}  + P^ \bot  u_1 (x)\overline {u_1 (y)} , \hfill \\
  \overline {\phi (x)} \phi (y) &=& \overline {P\phi (x)} P\phi (y) + \overline {P\phi (x)} P^ \bot  \phi (y) + \overline {P^ \bot  \phi (x)} \phi (y).
  \eqq
Note that all terms involved with either $P^{\bot}u_1$ or $P^{\bot}\phi$ have negligible contribution. For example, we have
\[
\begin{gathered}
  \left| {\iint\limits_{\mathbb{R}^3  \times \mathbb{R}^3 } {\frac{{Pu_1 (x)\overline {Pu_1 (y)} \overline {P^ \bot  \phi (x)} \phi (y)}}
{{|x - y|}}}} \right| \leqslant 2||Pu_1 ||^2 ||P^ \bot  \phi ||.||\nabla P\phi || \leqslant CZ^{3/2} , \hfill \\
  \left| {\iint\limits_{\mathbb{R}^3  \times \mathbb{R}^3 } {\frac{{P^ \bot  u_1 (x)\overline {u_1 (y)} \overline {P\phi (x)} P\phi (y)}}
{{|x - y|}}}} \right| \leqslant 2||P^ \bot  u_1 ||.||u_1 ||.||P\phi ||.||\nabla P\phi || \leqslant CZ^{2 - 1/4 + \varepsilon /2} . \hfill \\ 
\end{gathered} 
\]
Thus
\bq \label{eq:pert-quadratic-u1}
 &~& \lambda _1 \iint\limits_{\mathbb{R}^3  \times \mathbb{R}^3 } {\frac{{u_1 (x)\overline {u_1 (y)} \overline {\phi (x)} \phi (y)}}
{{|x - y|}}} - \sqrt {\lambda _1 (1 + \lambda _1 )} \operatorname{Re} \iint\limits_{\mathbb{R}^3  \times \mathbb{R}^3 } {\frac{{u_1 (x)u_1 (y)\overline {\phi (x)} \overline {\phi (y)} }}
{{|x - y|}}} \nn\hfill \\
   &\geqslant & (\lambda _1  - \sqrt {\lambda _1 (1 + \lambda _1 )} )\iint\limits_{\mathbb{R}^3  \times \mathbb{R}^3 } {\frac{{u_1 (x)\overline {u_1 (y)} \overline {\phi (x)} \phi (y)}}
{{|x - y|}}} \nn\hfill \\
   &\geqslant & (\lambda _1  - \sqrt {\lambda _1 (1 + \lambda _1 )} )\iint\limits_{\mathbb{R}^3  \times \mathbb{R}^3 } {\frac{{Pu_1 (x)\overline {Pu_1 (y)} \overline {P\phi (x)} P\phi (y)}}
{{|x - y|}}} - CZ^{2 - 1/4 + \varepsilon /2}. 
\eq

Next, consider the terms involved with $(\gamma',\alpha')$. In the integral, 
\[
{\iint\limits_{\mathbb{R}^3  \times \mathbb{R}^3 } {\frac{{\gamma '(x,y)\overline {\phi (x)} \phi (y)}}
{{|x - y|}}}}
\]
we use the decomposition
\bqq
  \gamma ' &=& P^ \bot  \gamma 'P^ \bot   + P^ \bot  \gamma 'P + P\gamma ' ,\hfill \\
  \overline {\phi (x)} \phi (y) &=& \overline {P\phi (x)} P\phi (y) + \overline {P^ \bot  \phi (x)} P\phi (y) + \overline {P^ \bot  \phi (x)} \phi (y) .
  \eqq 
Observe that all terms involved with either $P\gamma '$ or $P^ \bot  \phi$ have negligible contribution. For example, we have 
\bqq
  \left| {\iint\limits_{\mathbb{R}^3  \times \mathbb{R}^3 } {\frac{{(P\gamma ')(x,y)\overline {P\phi (x)} P\phi (y)}}
{{|x - y|}}}} \right| &\leqslant & 2[\operatorname{Tr} (P(\gamma ')^2 P)]^{1/2} ||P\phi ||.||\nabla P\phi || \hfill \\
   &\leqslant & 2[\operatorname{Tr} (\gamma ')]^{1/2} [\operatorname{Tr} (P\gamma ')]^{1/2} ||P\phi ||.||\nabla P\phi ||   \leqslant CZ^{2 - 1/8 + \varepsilon /4} 
   \eqq
and
\[
\left| {\iint\limits_{\mathbb{R}^3  \times \mathbb{R}^3 } {\frac{{(P^ \bot  \gamma 'P^ \bot  )(x,y)\overline {P^ \bot  \phi (x)} \phi (y)}}
{{|x - y|}}}} \right| \leqslant 2[\operatorname{Tr} (P^ \bot  (\gamma ')^2 P^ \bot  )]^{1/2} ||P^ \bot  \phi ||.||\nabla P\phi || \leqslant CZ^{3/2} .
\]
Thus
\bq \label{eq:pert-quadratic-gamma'}
\iint\limits_{\mathbb{R}^3  \times \mathbb{R}^3 } {\frac{{\gamma '(x,y)\overline {\phi (x)} \phi (y)}}
{{|x - y|}}} \geqslant \iint\limits_{\mathbb{R}^3  \times \mathbb{R}^3 } {\frac{{(P^ \bot  \gamma 'P^ \bot  )(x,y)\overline {P\phi (x)} P\phi (y)}}
{{|x - y|}}} - CZ^{2 - 1/8 + \varepsilon /4} .
\eq
Similarly we have
\bq \label{eq:pert-quadratic-alpha'}
\iint\limits_{\mathbb{R}^3  \times \mathbb{R}^3 } {\frac{{\alpha '(x,y)\overline {\phi (x)} \phi (y)}}
{{|x - y|}}} \geqslant \iint\limits_{\mathbb{R}^3  \times \mathbb{R}^3 } {\frac{{(P^ \bot  \alpha 'P^ \bot  )(x,y)\overline {P\phi (x)} P\phi (y)}}
{{|x - y|}}} - CZ^{2 - 1/8 + \varepsilon /4} .
\eq

Putting (\ref{eq:pert-quadratic-u1}), (\ref{eq:pert-quadratic-gamma'}), (\ref{eq:pert-quadratic-alpha'}) together and using the fact $h_{t,Z}\ge 0$ and $h_{t,Z}P=0$, we obtain
\[
q_{t,Z} (\gamma ,\alpha ,\phi ) \geqslant \frac{{\left\| {P\phi } \right\|^2 }}
{{||\phi _{t,Z} ||^2 }} q_{t,Z} (\gamma '',\alpha '') - CZ^{2 - 1/8 + \varepsilon /2}  \geqslant \frac{{\left\| {P\phi } \right\|^2 }}
{{||\phi _{t,Z} ||^2 }} Z^2 \mu (t) - CZ^{2 - 1/8 + \varepsilon /4} .
\]

4. In summary, from Case 1 and Case 2 we have in any case
\[
q_{t,Z} (\gamma ,\alpha ,\phi ) \geqslant \frac{{\left\| {P\phi } \right\|^2 }}
{{||\phi _{t,Z} ||^2 }} Z^2 \mu (t) - C\max \{ Z^{2 - \varepsilon } ,Z^{2 - 1/8 + \varepsilon /4} \} .
\]
Choosing $\eps=1/10$ we obtain
\[
\mathop {\inf }\limits_{(\gamma ,\alpha ) \in \G^{\rm B} } q_{t,Z} (\gamma ,\alpha ,\phi ) \geqslant \frac{{\left\| {P\phi } \right\|^2 }}
{{||\phi _{t,Z} ||^2 }} Z^2 \mu (t) - CZ^{2 - 1/10} .
\]
\end{proof}  

\subsection{Bogoliubov ground state energy}

We are now ready to give the proof of Theorem \ref{thm:GSE-Bogoluibov-bosonic-atom}.

\begin{proof} {\bf Upper bound.} Fix $\eps>0$ small. Choose $(\gamma_{t,\eps},\alpha_{t,\eps})\in \G^{\rm B}$  such that $$q_t(\gamma_{t,\eps},\alpha_{t,\eps})\le \mu(t)+\eps.$$
Choosing$\gamma(x,y)=Z^3\gamma_{t,\eps}(Zx,Zy)$, $\alpha(x,y)=Z^3\alpha_{t,\eps}(Zx,Zy)$ and $\phi(x)= Z^2\phi_{t-\Tr(\gamma_{t,\eps})/Z} (Zx)$, we have $\Tr(\gamma)+||\phi||^2=tZ$ and
\bqq
  \E^{\rm B} (\gamma ,\alpha ,\phi ,Z) &=& Z^3 E_{Z = 1}^H (\phi _{t - \Tr(\gamma _{t,\eps} )/Z} ) + Z^2 \left[ {q_t (\gamma _{t,\eps} ,\alpha _{t,\eps} ) + e'(t)\Tr(\gamma _{t,\eps} )} \right] \hfill \\
   &~&~+ Z\left[ {D(\rho _{\gamma _{t,\eps} } ,\rho _{\gamma _{t,\eps} } ) + X(\gamma _{t,\eps} ,\gamma _t ) + X(\alpha _{t,\eps} ,\alpha _{t,\eps} )} \right] \hfill \\
   &\le & Z^3 e(t - \Tr(\gamma _{t,\eps} )/Z) + Z^2 [\mu (t) + e'(t)\Tr(\gamma _{t,\eps} )] + Z C_\eps \hfill \\
   &=& Z^3 \left[ {e(t) - (\Tr(\gamma _{t,\eps})/Z)e'(t) + o (Z^{ - 1} )\Tr \gamma_\eps} \right] \hfill\\
&~&+ Z^2 \left[ {\mu (t)+\eps + \Tr(\gamma _{t,\eps} )e'(t)} \right] +  Z C_\eps \hfill \\
   &=& Z^3 e(t) + Z^2 (\mu (t)+\eps + o(1) C_\eps ).
 \eqq
Thus
$$ E^{\rm B}(N,Z)\le  Z^3 e(t) + Z^2 (\mu (t)+\eps + o(1) C_\eps ).$$
Because $\eps>0$ can be chosen as small as we want, we can conclude that
$$ E^{\rm B}(N,Z)\le Z^3 e(t) + Z^2 \mu (t) + o(Z^2 ).$$
 
{\bf Lower bound.} It suffices to consider $(\gamma,\alpha,\phi)$ such that $\E^{\rm B} (\gamma ,\alpha ,\phi,Z )\le Z^3e(t)$, and hence $\Tr[-\Delta \widetilde\gamma]\le CZ^3$. We shall denote by $P$ the one-dimensional projection onto the Hartree ground state $\phi_{t,Z}$ and $P^{\bot}=1-P$.

In the expression of $\E^{\rm B}(\gamma,\alpha,\phi,Z)$, if we ignore the non-negative terms $X(\gamma,\gamma)$, $X(\alpha,\alpha)$ and estimate the direct term by 
\bqq
  D(\rho _{\widetilde\gamma } ,\rho _{\widetilde\gamma } )&=&2D(\rho _{\widetilde\gamma } ,|\phi _{t,Z} |^2 ) - D(|\phi _{t,Z} |^2 ,|\phi _{t,Z} |^2 )+D(\rho _{\widetilde\gamma }-|\phi_{t,Z}|^2,\rho _{\widetilde\gamma }-|\phi_{t,Z}|^2)\hfill\\
 &\geqslant& 2D(\rho _{\widetilde\gamma } ,|\phi _{t,Z} |^2 ) - D(|\phi _{t,Z} |^2 ,|\phi _{t,Z} |^2 ) = 2D(\rho _{\widetilde\gamma } ,|\phi _{t,Z} |^2 )  + Z^3 e(t)- Z^2 e'(t)\Tr(\widetilde\gamma )
   \eqq
then we arrive at 
\bq \label{eq:E>=H+0q} E^{\rm B}(\gamma ,\alpha ,\phi,Z ) \geqslant Z^3 e(t)+\Tr( h_{t,Z} \widetilde\gamma) + \iint\limits_{\mathbb{R}^3  \times \mathbb{R}^3 } {\frac{{\gamma (x,y) \overline{\phi (x)} {\phi (y)}}}
{{|x - y|}}} + \operatorname{Re} \iint\limits_{\mathbb{R}^3  \times \mathbb{R}^3 } {\frac{{\alpha (x,y)\overline{\phi (x)}\overline{\phi (y)}}}
{{|x - y|}}}.
\eq 

By the same argument of the proof of Lemma \ref{le:mu-t-finite} we have
\[
 \frac{1}{2} \operatorname{Tr} [h_{t,Z} \gamma ] + \iint\limits_{\mathbb{R}^3  \times \mathbb{R}^3 } {\frac{{\gamma (x,y) \overline{\phi (x)} {\phi (y)}}}
{{|x - y|}}} + \operatorname{Re} \iint\limits_{\mathbb{R}^3  \times \mathbb{R}^3 } {\frac{{\alpha (x,y)\overline{\phi (x)}\overline{\phi (y)}}}
{{|x - y|}}} \ge -CZ^2 .
\]
Putting this bound together with the gap $\operatorname{Tr} [h_{t,Z} \widetilde\gamma ] \ge \Delta_t Z^2 \Tr(P^{\bot}\widetilde \gamma)$ into (\ref{eq:E>=H+0q}), and comparing with the upper bound $\E^{\rm B} (\gamma ,\alpha ,\phi,Z )\le Z^3e(t)$ we obtain $||P^{\bot}\phi||\le C$.

We are now able to apply Lemma \ref{le:pert-quadratic} to conclude from (\ref{eq:E>=H+0q}) that
\bqq
E^{\rm B}(\gamma ,\alpha ,\phi ,Z ) \geqslant Z^3 e(t) + (\phi, h_{t,Z} \phi)+ \frac{{\left\| {P\phi } \right\|^2 }}
{{||\phi _{t,Z} ||^2 }} Z^2\mu (t) -CZ^{2-1/10}.
\eqq
Because ${\left\| {P\phi } \right\|^2} \le {\left\| {\phi_{t,Z} } \right\|^2} = tZ $, we obtain the desired lower bound. 
\end{proof}

From the above proof of the lower bound, we also obtain the following estimates on the ground state, which will be useful in the proof of the binding up to the critical number $t_cZ$.

\begin{lemma}[Properties of Bogoliubov minimizers]\label{le:estimate-ground-state} If $(\gamma,\alpha,\phi)$ is a minimizer for $E^{\rm B}(N,Z)$ (or more generally, if $E^{\rm B} (\gamma ,\alpha ,\phi ,Z)=Z^3 e(t)+ Z^2 \mu (t)+o(Z^2)$) then $\Tr(P^ \bot \widetilde\gamma) \le  C$, $\left\langle {\phi ,{h_{t,Z}}\phi } \right\rangle  = o({Z^2})$ and 
\bqq
  D({\rho _{\widetilde \gamma }} - |{\phi _{t,Z}}{|^2},{\rho _{\widetilde \gamma }} - |{\phi _{t,Z}}{|^2}) = o({Z^2}).
\eqq
In particular, it follows from $\left\langle {\phi ,{h_{t,Z}}\phi } \right\rangle  = o({Z^2})$  that ${\left\| {P\phi } \right\|^2} = tZ + o(Z)$. Here $P$ is the one-dimensional projection onto the Hartree ground state $\phi_{t,Z}$.
\end{lemma}

\subsection{Comparison to quantum energy: a heuristic discussion}

Let us discuss on the comparison between the Bogoliubov ground state energy $E^{\rm B}(N,Z)$ and the quantum energy $E(N,Z)$ in Conjecture \ref{conj:bosonic-atom-energy}.

First at all, due to the variational principle, the Bogoliubov energy $E^{\rm B}(N,Z)$ is a rigorous upper bound to the quantum grand canonical energy
$$
E^{\rm g}(N,Z)=\inf \{(\Psi, \bigoplus_{N=0}^\infty H_{N,Z} \Psi), \Psi\in \F, ||\Psi||=1\}.
$$
It is believed that the ground state energy $E(N,Z)$ is a convex function on $N$ (see \cite{LS10}, p. 229), which is equivalent to $E^{\rm g}(N,Z)=E(N,Z)$. If this conjecture is correct then the Bogoliubov energy $E^{\rm B}(N,Z)$ is also an upper bound to the canonical energy $E(N,Z)$.

In the following, we shall argue heuristically why the Bogoliubov energy $E^{\rm B}(N,Z)$ is a lower bound to $E(N,Z)$ (up to an error $o(Z^2)$). Some further work is required to make the argument rigorous.  

Choosing an orthonormal basis $\{u_n\}_{n=0}^\infty$ for $\h$ with $u_0=\phi_{t,Z}/||\phi_{t,Z}||$, we can represent the Hamiltonian $\mathbb{H}_Z=\bigoplus_{N=0}^\infty H_{N,Z}$ in the second quantization
$$ \mathbb{H}_Z=\sum_{m,n\ge 0} h_{m,n}a^*_m a_n+\frac{1}{2} \sum_{m,n,p,q\ge 0}W_{m,n,p,q} a^*_m a^*_n a_p a_q $$
where $a_n=a(u_n)$ and  
$$h_{m,n}=(u_m, (-\Delta-Z|x|^{-1})u_n), W_{m,n,p,q}=\iint_{\R^3 \times \R^3} {\frac{\overline{u_m(x)}\overline{u_n(y)} u_p(x) u_q(y)}{|x-y|}}.$$
Assume that $\Psi$ is a ground state for $E(N,Z)$. We shall denote by $\left\langle {\mathbb{H}_Z} \right\rangle_\Psi$ the expectation $\left\langle {\Psi, \mathbb{H}_Z \Psi} \right\rangle$.

{\bf Step 1.} As in \cite{BLLS93} we have the condensation $\Tr (P^{\bot} \gamma_{\Psi})\le C$ where $P$ is the one-dimensional projection onto $u_0$. Let us denote $\gamma =P^{\bot}\gamma_\Psi P^{\bot}$, $\alpha=P^{\bot}\alpha_\Psi P^{\bot}$, $\N_0=a_0^*a_0$ and $N_0=\left\langle {\N_0} \right\rangle_\Psi$. Then $(\gamma,\alpha)\in \G^{\rm B}$ and $N-N_0=\Tr(\gamma)\le C$. 

{\bf Step 2.} The leading term $Z^3e(t)$ of the ground state energy $E(N,Z)$ comes from the terms of full condensation, namely ${h_{00}}a_0^*a_0$ and $W_{0000}a_0^*a_0^*a_0a_0$. Similarly to the computation to the energy of product functions, we have
\bq \label{eq:Bosonic-energy-condensation-part}
 &~& {h_{00}}\left\langle {a_0^*{a_0}} \right\rangle_\Psi  + {W_{0000}}\left\langle {a_0^*a_0^*{a_0}{a_0}} \right\rangle _\Psi \nn \hfill \\
   &=& \left\langle {{u_0},\left( { - \Delta  - Z|x{|^{ - 1}}} \right){u_0}} \right\rangle {N_0} + (\left\langle {\N_0^2 } \right\rangle _\Psi -N_0 ) D(|{u_0}{|^2},|{u_0}{|^2}) \nn\hfill \\
   &\ge & \left\langle {{u_0},\left( { - \Delta  - Z|x{|^{ - 1}}} \right){u_0}} \right\rangle {N_0} + (N_0^2-N_0) D(|{u_0}{|^2},|{u_0}{|^2}) \nn\hfill \\
      &\geqslant & \frac{{{N_0}{Z^3}}}{{{N_0} - 1}}e\left( {\frac{{{N_0} - 1}}{Z}} \right)\nn \hfill \\
   &=& {Z^3}e(t) - {Z^2}e'(t)\Tr (\gamma) + {Z^2}[{t^{ - 1}}e(t) - e'(t)] + o({Z^2}) .
 \eq
As a consequence, the expectation of the rest of the Hamiltonian $\mathbb{H}_Z$ should be of order $O(Z^2)$.

{\bf Step 3.} Because almost of particles live in the condensation $u_0$, we may hope to eliminate all terms $W_{m,n,p,q} a^*_m a^*_n a_p a_q $ in the two-body interaction which have only $0$ or $1$ operator $a^{\#}_0$ (where $a^{\#}_0$ is either $a_0$ or $a^{*}_0$).

{\bf Step 4.} Now we apply the Bogoliubov principle in which we replace any $a^{\#}_0$ by $\sqrt{N_0}\approx \sqrt{N}$. We can see that the terms with 1 and 3 operators $a^{\#}_0$ should be canceled together. In fact,
\bqq
 &~& \sum\limits_{m \geqslant 1} {\left( {{h_{0m}}{{\left\langle {a_0^*{a_m}} \right\rangle }_\Psi } + {W_{000m}}{{\left\langle {a_0^*a_0^*{a_0}{a_m}} \right\rangle }_\Psi }} \right)}  \hfill \\
   &\approx & \sum\limits_{m \geqslant 1} {\left( {{h_{0m}}{{\left\langle {a_0^*{a_m}} \right\rangle }_\Psi } + {W_{000m}}{{\left\langle {a_0^*{a_m}} \right\rangle }_\Psi }} \right)}  \hfill \\
   &=& \sum\limits_{m \geqslant 1} {\left\langle {{u_m},\left( { - \Delta  - Z|x{|^{ - 1}} + N|{u_0}|*|.{|^{ - 1}}} \right){u_0}} \right\rangle {{\left\langle {a_0^*{a_m}} \right\rangle }_\Psi }}  = 0 \hfill \\
   &=& \sum\limits_{m \geqslant 1} {\left\langle {{u_m}, {Z^2}e'(t){u_0}} \right\rangle {{\left\langle {a_0^*{a_m}} \right\rangle }_\Psi }}  = 0 .
   \eqq
Here we use the fact that $u_0$ is the ground state for the Hartree mean-field operator
$$ h_{t,Z}=-\Delta-Z|x|^{-1}+|\phi_{t,Z}|^2*|.|^{-1}-Z^2e'(t)$$   
It remains the terms with precisely 0 or 2 operators  $a^{\#}_0$,
\bq \label{eq:Bosonic-energy-quaratic-part-1}
 &~& \sum\limits_{m,n \geqslant 1} {\left( {{h_{mn}}{{\left\langle {a_m^*{a_n}} \right\rangle }_\Psi } + {W_{m00n}}{{\left\langle {a_m^*a_0^*{a_0}{a_n}} \right\rangle }_\Psi }} \right)} \nn \hfill \\
   &\approx & \sum\limits_{m,n \geqslant 1} {\left( {{h_{mn}}{{\left\langle {a_m^*{a_n}} \right\rangle }_\Psi } + N{W_{m00n}}{{\left\langle {a_m^*{a_n}} \right\rangle }_\Psi }} \right)}  \nn \hfill \\
   &=& \operatorname{Tr} \left[ {\left( { - \Delta  - Z|x{|^{ - 1}} + N|{u_0}{|^2}*|.{|^{ - 1}}} \right)\gamma } \right]
 \eq
 and
 \bq \label{eq:Bosonic-energy-quaratic-part-2}
  &~&\sum\limits_{m,n \geqslant 1} {\left( {{W_{m0n0}}{{\left\langle {a_m^*a_0^*{a_0}{a_n}} \right\rangle }_\Psi } + \operatorname{Re} [{W_{mn00}}{{\left\langle {a_m^*a_n^*{a_0}{a_0}} \right\rangle }_\Psi }]} \right)}  \nn \hfill \\
  &\approx & \sum\limits_{m,n \geqslant 1} {\left( {N{W_{m0n0}}{{\left\langle {a_m^*{a_n}} \right\rangle }_\Psi } + N\operatorname{Re} [{W_{mn00}}{{\left\langle {a_m^*a_n^*} \right\rangle }_\Psi }]} \right)} \nn \hfill \\
   &=& \operatorname{Re} \iint\limits_{{\mathbb{R}^3} \times {\mathbb{R}^3}} {\frac{{(\gamma (x,y) +  \alpha (x,y)){\phi _{t,Z}}(x){\phi _{t,Z}}(y)}}{{|x - y|}}dxdy}.
   \eq

{\bf Step 5.} Putting the approximations (\ref{eq:Bosonic-energy-condensation-part}), (\ref{eq:Bosonic-energy-quaratic-part-1}) and (\ref{eq:Bosonic-energy-quaratic-part-2}) together we obtain the desired lower bound
\bqq
  {\left\langle {{\mathbb{H}_Z}} \right\rangle _\Psi } &\geqslant & {Z^3}e'(t) + {Z^2}[{t^{ - 1}}e(t) - e'(t)] \hfill \\
   &~&+ \operatorname{Tr} [{h_{t,Z}}\gamma ] + \operatorname{Re}\iint\limits_{{\mathbb{R}^3} \times {\mathbb{R}^3}} {\frac{{[\gamma (x,y) +  \alpha (x,y)]{\phi _t}(x){\phi _t}(y)}}{{|x - y|}}} + o({Z^2}).
 \eqq
 Because $(\gamma,\alpha)\in \G^{\rm B}$ and $\gamma \phi_{t,Z}=0$ one has
 \[  {Z^2}[{t^{ - 1}}e(t) - e'(t)] + \operatorname{Tr} [{h_{t,Z}}\gamma ] +  \operatorname{Re}\iint\limits_{{\mathbb{R}^3} \times {\mathbb{R}^3}} {\frac{{[\gamma (x,y) + \alpha (x,y)]{\phi _{t,Z}}(x){\phi _{t,Z}}(y)}}{{|x - y|}}} \ge Z^2 \mu (t) .
\]
Thus we arrive at the desired lower bound
\[{\left\langle {{\mathbb{H}_Z}} \right\rangle _\Psi } \geqslant {Z^3}e'(t) + Z{}^2\mu (t) + o({Z^2}).\]
\section*{Appendix}
\addcontentsline{toc}{section}{Appendix}

\begin{proof}[Proof of Lemma \ref{le:relation-gamma-alpha}] It is obvious that $\Gamma\ge 0$ if and only if $\gamma\ge 0$, $\alpha^*=J\alpha J$ and 
\bqq \label{eq:Gamma-positive-a}
\left\langle {f \oplus Jg,\Gamma f \oplus Jg} \right\rangle  = (f,\gamma f) + (g,(1 + \gamma )g) + 2\operatorname{Re} (\alpha J f,g) \geqslant 0,~~\forall f,g\in \h.
\eqq
Using a simple scaling $g=tg$, $t\in \C$, we can see that the latter inequality is equivalent to 
$$
(f,\gamma f)(g,(1 + \gamma )g) \ge |(\alpha Jf,g)|^2,~~\forall f,g\in \h.
$$
Replacing $g$ by $(1+\gamma)^{-1}g$, we can rewrite the latter inequality as
\bq \label{eq:Gamma-positive-b}
(f,\gamma f)(g,(1 + \gamma )^{-1}g) \ge |(\alpha J f,(1+\gamma)^{-1}g)|^2,~~\forall f,g\in \h.
\eq

Note that (\ref{eq:relation-gamma-alpha}) follows from (\ref{eq:Gamma-positive-b}) by choosing $g =\alpha Jf$. Reversely, we can see that (\ref{eq:relation-gamma-alpha})  implies (\ref{eq:Gamma-positive-b}) by using  the Cauchy-Schwarz inequality for the positive definite quadratic form $Q(u,v)=(u,(1+\gamma)^{-1}v)$, i.e.  
\[
(f,\gamma f)(g,(1 + \gamma )^{ - 1} g) \geqslant (\alpha Jf,(1 + \lambda )^{ - 1} \alpha Jf)(g,(1 + \gamma )^{ - 1} g) \geqslant \left| {(\alpha Jf,(1 + \gamma )^{ - 1} g)} \right|^2.
\]
\end{proof}

\begin{proof}[Proof of Theorem \ref{thm:unitary-implementation}] The proof below follows \cite{So07} Theorem 9.5 (sufficiency) and \cite{Ru77} Theorem 6.1 (necessity). 

{\bf Sufficiency.} Assume that $VV^*$ is trace class on $\h$. We shall construct the unitary $\U_\V$. 

1. Let $\{u_i\}_{i\ge 1}$ be an orthonormal basis for $\h$. Recall that an orthonormal basis for $\F^{B,F}(\h)$ is given by
\[
\left| {n_{i_1 } ,...,n_{i_M } } \right\rangle  = \left( {n_{i_1 } !...n_{i_M } !} \right)^{ - 1/2} a ^* (u_{i_M } )^{n_{i_M } } ...a^* (u_{i_1 } )^{n_{i_1 } } \left| 0 \right\rangle ,
\]
where $n_j$ run over $0,1,2,...$ such that there are only finite $n_j>0$. 

We start by constructing the new vacuum $\left| 0 \right\rangle_\V=\U_\V\left| 0 \right\rangle$ which is characterized by
\bqq
A(\V(u_i\oplus 0))\left| 0 \right\rangle_\V=0.
\eqq
for all $i=1,2,...$, namely 
$$A(\V (u_i\oplus 0))=A(Uu_i\oplus JVJu_i)=a(Uu_i)+a^*(VJ u_i)$$
are the new annihilation operators.

2. The first step is to choose an convenient basis $\{u_i\}$. From $\V^*\S\V=\S=\V\S\V^*$ we have 
$$UU^*=1+VV^*,~U^*U=1+J^*V^*VJ$$
and $C=C^*$ where $C=U^*VJ$. Since $U^*U-1$ is trace class, $U^*U$ has an orthonormal eigenbasis on $\h$. On the other hand, because $U^*U$ commutes with the conjugate linear map $C=U^*J^*V$ and $C^*C$ is trace class on $\h$, we can find an orthonormal basis $\{u_i\}_{i\ge 1}$ for $\h$ consisting of eigenvectors of $U^*U$  such that they are also eigenvectors of $C$. 

Denote $\mu_i:=||Uu_i||\ge 1$ and $f_i:=\mu_i^{-1}Uu_i$. Then $\{f_i\}_{i\ge 1}$ is an orthonormal basis for $\h$. Since  
\[
(f_j , VJ u_i ) = \mu_j^{ - 1} (Uu_j , VJ u_i ) = \mu_j^{ - 1} (u_j ,Cu_i ) = 0 
\]
for all $j\ne i$, we must have $VJ u_i\in \Span\{f_i\}$. Note that if we change $u_i$'s by complex phases then it still holds that $(u_j ,Cu_i ) =0$ for all $i\ne j$ (although  $u_i$'s maybe no longer eigenvectors of $C$). Therefore, we can change $u_i$'s by complex phases to obtain $VJ u_i=\nu_i f_i$ for some $\nu_i\ge 0$. Thus there is an orthonormal basis $\{f_i\}_{i\ge 1}$ for $\h$ such that the new annihilation operators are 
\bqq
A(\V (u_i \oplus 0))=\mu_i a(f_i)+\nu_i a^*(f_i), ~~~i=1,2,...
\eqq
where $\mu_i\ge 1$, $\nu_i\ge 0$, $\mu _i^2-\nu_i^2=1$ and $\sum_{i=1}\nu_i^2=\Tr(VV^*)<\infty$.

3. This representation allows us to construct the new vacuum $\left| 0 \right\rangle_\V$ explicitly
\bqq
  \left| 0 \right\rangle_\V &=& \mathop {\lim }\limits_{M \to \infty } \prod\limits_{j = 1}^M {(1 - (\nu _j /\mu _j )^2 )^{1/4} \sum\limits_{n = 0}^\infty  {\left( { - \frac{{\nu _j }}
{{2\mu _j }}} \right)^n \frac{{a^* (f_j )^{2n} }}
{{n!}}} } \left| 0 \right\rangle \nn \hfill \\
   &=& \prod\limits_{j = 1}^{} {(1 - (\nu _j /\mu _j )^2 )^{1/4} \exp \left[ { - \sum\limits_{i = 1} {\frac{{\nu _i }}
{{2\mu _i }}a ^* (f_i )^2 } } \right]} \left| 0 \right\rangle.
\eqq
It is straightforward to check that $\left| 0 \right\rangle_\V$ is well defined and is annihilated by the new annihilation operators $A(\V (u_i\oplus 0))$. Having the new vacuum $\left| 0 \right\rangle_\V$, we can define $\left| n_{i_1},...,n_{i_M} \right\rangle_\V=\U \left| n_{i_1},...,n_{i_M} \right\rangle$ by
$$
\left| n_{i_1},...,n_{i_M} \right\rangle_\V=(n_{i_1}!...n_{i_M}!)^{-1/2} A^*(\V(u_{i_M}\oplus 0))^{n_{i_M}}...A^*(\V(u_{i_1}\oplus 0))^{n_{i_1}}\left| 0 \right\rangle_\V.
$$

4. Finally we need to prove that the new vectors $\left| n_{i_1},...,n_{i_M} \right\rangle_\V$ indeed form a basis for $\F$. The trick is to use the formula 
\bqq
  \left| 0 \right\rangle = \prod\limits_{j = 1} {(1 - (\nu _j /\mu _j )^2 )^{ - 1/4} \exp \left[ {\sum\limits_{i = 1} {\frac{{\nu _j }}
{{2\mu _j }}a_ + ^* (f_i )^2 } } \right]} \left| 0 \right\rangle_\V.
\eqq
and express the old basis vectors $\left| n_{i_1},...,n_{i_M} \right\rangle$ in terms of the new ones. Since the new vectors $\left| n_{i_1},...,n_{i_M} \right\rangle_\V$ span all of the old basis vectors $\left| n_{i_1},...,n_{i_M} \right\rangle$, the new ones span the whole space $\F$. 

{\bf Necessity.} Assume that there exists a normalized vector $\left| 0 \right\rangle _\V\in \F$ such that $A (V(u\oplus 0))\left| 0 \right\rangle _\V  = 0$ for all $u\in \h$. We shall prove that $VV^*$ must be trace class on $\h$. 

5. Let $\left| 0 \right\rangle _\V  =  \mathop\bigoplus_{N=0}^{\infty}\Psi_N$ where $\Psi _N  \in  \mathop\bigotimes_{\text{sym}}^N \h$. Then the condition $A(V(u\oplus 0))\left| 0 \right\rangle _\V  = 0$ is equivalent to
\bq
a(Uu)\Psi_1=0~\text{and }~a (Uu)\Psi _{N+2}  + a^* (VJ u)\Psi _{N}  = 0~~\text{for all} ~u\in \h,N=0,1,2,...   \label{eq:Shale-Stinespring-necessity-vanish}
\eq

Since $UU^*=1+VV^*\ge 1$ we have $\Ker(U^*)=\{0\}$, and hence $\overline {\text{Ran}(U)}=\h$. Therefore, it follows from $a(Uu)\Psi_1=0$ for all $u\in \h$ that $\Psi_1=0$. Then, by induction using (\ref{eq:Shale-Stinespring-necessity-vanish}) we obtain $\Psi_1=\Psi_3=\Psi_5=...=0$. 

If $\Psi_0=0$ then the same argument deduces $\Psi_0=\Psi_2=\Psi_4=...=0$ which contradicts with $\left| 0 \right\rangle _\V\ne 0$. Thus $\Psi_0\in \mathbb{C}\backslash\{0\}$ and from (\ref{eq:Shale-Stinespring-necessity-vanish}) with $N=0$ we have
\bq
a (Uu)\Psi _2+ \Psi _0VJ u =  0~~\text{for all}~u\in\h. \label{eq:Shale-Stinespring-necessity-N2}
\eq

6. Introducing the conjugate linear map $H:\h\to\h$ defined by
\[
(H\varphi _1,\varphi _2 ) = \left(\Psi _2, {\varphi _1  \otimes \varphi _2} \right)~~\text{for all}~ \varphi _1,\varphi _2\in \h.
\]
A  straightforward computation shows that $\Tr(H^* H) =\left\| {\Psi _2 } \right\|^2$. Moreover using (\ref{eq:Shale-Stinespring-necessity-N2}) and the symmetry of $\Psi_2$ we have 
\bqq
  (- \Psi _0VJ\varphi _1 , \varphi _2 ) &=& (a(U\varphi _1 )\Psi _2 ,\varphi _2) = (\Psi _2, a^* (U\varphi_1 )\varphi _2,) \hfill \\
   &=& \sqrt{2} (\Psi _2,U\varphi _1  \otimes \varphi _2) =\sqrt{2}  (H U\varphi _1 ,\varphi _2)
\eqq
for all $\varphi_1,\varphi_2\in \h$. This means $- \Psi _0VJ=HU$. Because $U$ is bounded and $H^*H$ is trace class on $\h$, we conclude that 
$$VV^*=2\Psi_0^{-2}HUU^*H^*$$
is trace class on $\h$.
\end{proof}

\begin{proof} [The rest part of proof of Theorem \ref{thm:quasi-free-state}]

We now prove that $\Gamma$ is the 1-pdm of the state $\rho=\Tr[G]^{-1}G$. Recall that $\F$ has the orthonormal basis 
\[
\left| {n_1 ,n_2 ,...} \right\rangle  = (n_1 !n_2 !...)^{ - 1/2} (a_1^* )^{n_1 } (a_2^* )^{n_2 } ...\left| 0 \right\rangle 
\]
where $\left| 0 \right\rangle$ is the vacuum and $n_1,n_2...$ run over $0,1,2,...$ such that there are only finite $n_j>0$. A  straightforward computation shows that
 \begin{eqnarray*}
\Tr (G) &=& \sum\limits_{n_j  = 0,1,2,...} {\left\langle {n_1 ,n_2 ,...} \right| G \left| {n_1 ,n_2 ,...} \right\rangle }  \hfill \\
  &=& \sum\limits_{n_j=0,1,...; j\in I} {(n_1 !n_2 !...)^{ - 1} \left\langle 0 \right|\prod\limits_{i\in I} {\left( {a_i ^{n_i } \exp [- \lambda_i a_i^* a_i ](a_i^*)^{n_i} } \right)} \left| 0 \right\rangle }  \hfill \\
   &=& \sum\limits_{n_j=0,1,...; j\in I} {(n_1 !n_2 !...)^{ - 1} \left\langle 0 \right|\prod\limits_{i \in I} {\left( {a_i ^{n_i } \sum\limits_{k =0}^\infty  {\frac{{( -\lambda _i )^k (a_i^* a_i )^k }}
{{k!}}} (a_i^* )^{n_i } } \right)} \left| 0 \right\rangle }  \hfill \\
   &=& \sum\limits_{n_j=0,1,...; j\in I} {(n_1 !n_2 !...)^{ - 1} \left\langle 0 \right|\prod\limits_{i \in I} {\left( {\sum\limits_{k = 0}^\infty  {\frac{{( - e_i )^k (n_i )^k (n_i !)}}
{{k!}}} } \right)} \left| 0 \right\rangle }  \hfill \\
   &=& \sum\limits_{n_j=0,1,...; j\in I} {\prod\limits_{i}{ {e^{ - \lambda _i n_i }}}}  =\prod\limits _{i \in I} {\frac{1}
{{1 - e^{- \lambda _i}}}} < \infty
\end{eqnarray*}
since $\sum_{i\in I} e^{- \lambda _i}<\infty$. Thus $\rho$ is well-defined. 

We check that $\Gamma$ is indeed the 1-pdm of $\rho$. Note that $\left| {n_1 ,n_2 ,...} \right\rangle$ and $G\left| {n_1 ,n_2 ,...} \right\rangle$ have the same number of particle $u_i$ for any $i=1,2,...$. By the same way of determining $\operatorname{Tr} (G)$ we find that $
\operatorname{Tr} (a_i a_j G) = 0$ and 
\bqq
  &~&\operatorname{Tr} (a_i^* a_j G) =\delta _{ij} \operatorname{Tr} (a_i^* a_i G) \hfill \\
   &=& \delta _{ij} \left( {\prod\limits_{k \in \operatorname{I} ,k \ne i} {(1 + \lambda _k )} } \right)\left( {\sum\limits_{n_i  = 0}^\infty  {(n_i !)^{ - 1} \left\langle 0 \right|a_i^{n_i } a_i^* a_i \exp ( - c_i a_i^* a_i )(a_i^* )^{n_i } \left| 0 \right\rangle } } \right) \hfill \\
   &=& \delta _{ij} \left( {\prod\limits_{k \in \operatorname{I} ,k \ne i} {(1 + \lambda _k )} } \right)\left( {\sum\limits_{n_i  = 0}^\infty  {(n_i !)^{ - 1} \left\langle 0 \right|a_i^{n_i } a_i^* a_i \sum\limits_{r = 0}^\infty  {\frac{{( - c_i )^r (a_i^* a_i )^r }}
{{r!}}} (a_i^* )^{n_i } \left| 0 \right\rangle } } \right) \hfill \\
   &=& \delta _{ij} \left( {\prod\limits_{k \in \operatorname{I} ,k \ne i} {(1 + \lambda _k )} } \right)\left( {\sum\limits_{n_i  = 0}^\infty  {(n_i !)^{ - 1} \left\langle 0 \right|\sum\limits_{r = 0}^\infty  {\frac{{( - c_i )^r (n_i )^{r + 1} (n_i !)}}
{{r!}}} \left| 0 \right\rangle } } \right) \hfill \\
   &=& \delta _{ij} \left( {\prod\limits_{k \in \operatorname{I} ,k \ne i} {(1 + \lambda _k )} } \right)\left( {\sum\limits_{n_i  = 0}^\infty  {\exp ( - c_i n_i )n_i } } \right) = \delta _{ij} \lambda _i \prod\limits_{k \in \operatorname{I} } {(1 + \lambda _k )} 
\eqq
in which we have used
\bqq
  \sum\limits_{n_i  = 0}^\infty  {\exp ( - c_i n_i )n_i }  &=&  - \frac{d}
{{dc_i }}\sum\limits_{n_i  = 0}^\infty  {\exp ( - c_i n_i )}  =  - \frac{d}
{{dc_i }}\frac{1}
{{1 - \exp ( - c_i )}} \hfill \\
   &=& \frac{{\exp ( - c_i )}}
{{(1 - \exp ( - c_i ))^2 }} = \lambda _i (1 + \lambda _i ).
\eqq
From the above computations we find that 
\[
\rho (a_i a_j ) = \left( {\operatorname{Tr} (G)} \right)^{ - 1} \operatorname{Tr} (a_i a_j G) = 0 = (u_i ,\alpha Ju_j )
\]
and
\[
\rho (a_i^* a_j ) = \left( {\operatorname{Tr} (G)} \right)^{ - 1} \operatorname{Tr} (a_i^* a_j G) = \delta _{ij} \lambda _i  = (u_i ,\gamma u_j )
\]
for any $i,j$. Thus $\Gamma$ is indeed the 1-pdm of $\rho$. 

3. Finally, we check that $\rho$ is a quasi-free state. One way to do it is to consider $\rho$ as a limit of appropriate Gibbs states, see \cite{BLS94} (eq. (2b.34)). In the following, we shall give a more direct approach by mimicking the proof of Wick's Theorem in \cite{Ga60}.

It suffices to prove  (\ref{eq:Wick-odd})-(\ref{eq:Wick-even}) when $A(F_i)$ is either a creation or annihilation operator, which we denote by $c_i$. Our aim is to show that
\begin{eqnarray}\label{eq:Wick-induction}
 &~&{\mathop{\rm Tr}\nolimits} [c_1 c_2 c_3 c_4 ...c_k G] = \frac{{{\mathop{\rm Tr}\nolimits} [c_1 c_2 G]}}{{{\mathop{\rm Tr}\nolimits} [G]}}{\mathop{\rm Tr}\nolimits} [c_3 c_4 ...c_k G] \hfill\\ 
  &~&~~~~~~~+ \frac{{{\mathop{\rm Tr}\nolimits} [c_1 c_3 G]}}{{{\mathop{\rm Tr}\nolimits} [G]}}{\mathop{\rm Tr}\nolimits} [c_2 c_4 ...c_k G] + ... + \frac{{{\mathop{\rm Tr}\nolimits} [c_1 c_k G]}}{{{\mathop{\rm Tr}\nolimits} [G]}}{\mathop{\rm Tr}\nolimits} [c_2 c_3 ...c_{k - 1} G] \nonumber
 \end{eqnarray}
and the result follows immediately by a simple induction. By the same way of computating $\Tr[G]$
 we may check that
\begin{equation}\label{eq:Tr-ccG}
\frac{{{\mathop{\rm Tr}\nolimits} [c_1 c_2 G]}}{{{\mathop{\rm Tr}\nolimits} [G]}} = f(c_1 )[c_1 ,c_2 ]
\end{equation}
where $[c_1,c_2]=c_1c_2-c_2c_1\in \{0,-1,1\}$ and
\begin{equation}\label{eq:Tr-ccG-f}
f(c_1 ) = \left\{ \begin{array}{l}
 (1 - e^{-\lambda _j } )^{ - 1} {\rm ~if~} c_1=a_j, ~j\in I, \\ 
 (1 - e^{\lambda _j } )^{ - 1} ~~{\rm ~if~} c_1=a_j^*, ~j\in I, \\ 
 1 ~~~~~~~~~~~~~~~{\rm ~if~} c_1=a_j, ~j\notin I,\\ 
 0 ~~~~~~~~~~~~~~~{\rm ~if~} c_1=a_j^*, ~j\notin I.\\ 
 \end{array} \right.
\end{equation}
Thus (\ref{eq:Wick-induction}) is equivalent to 
\begin{eqnarray}\label{eq:Wick-induction-new}
 &~&{\mathop{\rm Tr}\nolimits} [c_1 c_2 c_3 c_4 ...c_k G] = f(c_1 )[c_1 ,c_2 ]{\mathop{\rm Tr}\nolimits} [c_3 c_4 ...c_k G] \hfill\\ 
  &~&~~~~~~~+ f(c_1 )[c_1 ,c_3 ]{\mathop{\rm Tr}\nolimits} [c_2 c_4 ...c_k G] + ... + f(c_1 )[c_1 ,c_k ]{\mathop{\rm Tr}\nolimits} [c_2 c_3 ...c_{k - 1} G] \nonumber .
 \end{eqnarray}

We can prove (\ref{eq:Wick-induction-new}) as follows. From the identity
$$
c_1 c_2 c_3 c_4 ...c_k=[c_1 ,c_2 ]c_3 c_4 ...c_k+ ... + c_2 c_4 ...c_{k - 1}[c_1 ,c_k ]+c_2 c_3 c_4 ...c_kc_1 
$$
we deduce that
\begin{eqnarray}
{\mathop{\rm Tr}\nolimits} \left[ {c_1 c_2 c_3 c_4 ...c_k G} \right] &=&{\mathop{\rm Tr}\nolimits} \left[ {[c_1 ,c_2 ]c_3 c_4 ...c_k G} \right] \hfill\\
&~&+ ... +{\mathop{\rm Tr}\nolimits} \left[ {c_2 c_4 ...c_{k - 1}[c_1 ,c_k ] G} \right]+{\mathop{\rm Tr}\nolimits} \left[ {c_2 c_3 c_4 ...c_kc_1 G} \right]\nonumber \label{eq:Wick-induction-1}.
 \end{eqnarray}

We first consider when $c_1$ is either $a_j$ or $a_j^*$ with $j\in I$. In this case it is straightforward to see that
$
c_1 G  = e^{ \pm \lambda _j } c_1 G
$
where (+) if $c_1=a_j^*$ and (-) if $c_1=a_j$. This implies that 
\begin{eqnarray}\label{eq:Wick-induction-2}
{\mathop{\rm Tr}\nolimits} \left[ {c_2 c_3 c_4 ...c_k c_1G} \right]= e^{ \pm \lambda _j } {\mathop{\rm Tr}\nolimits} \left[ {c_2 c_3 c_4 ...c_k Gc_1} \right] = e^{ \pm \lambda _j } {\mathop{\rm Tr}\nolimits} \left[ {c_1 c_2 c_3 c_4 ...c_k G} \right].
\end{eqnarray}
Substituting (\ref{eq:Wick-induction-2})  into (\ref{eq:Wick-induction-1}) we conclude that
\bqq
 {\mathop{\rm Tr}\nolimits} \left[ {c_1 c_2 c_3 c_4 ...c_k G} \right] &=& \frac{{[c_1 ,c_2 ]}}{{1 - e^{ \pm \lambda _j } }}{\mathop{\rm Tr}\nolimits} \left[ {c_3 c_4 ...c_k G} \right]  \hfill\\ 
&+& \frac{{[c_1 ,c_3 ]}}{{1 - e^{ \pm \lambda _j } }}{\mathop{\rm Tr}\nolimits} \left[ {c_2 c_4 ...c_k G} \right] + ... + \frac{{[c_1 ,c_k ]}}{{1 - e^{\pm \lambda _j } }}{\mathop{\rm Tr}\nolimits} \left[ {c_2 c_4 ...c_{k - 1} G} \right]
\eqq 
which is precisely the desired identity (\ref{eq:Wick-induction-new}).
 
If $c_1=a_j$ for some $j\notin I$ then 
$$
{\mathop{\rm Tr}\nolimits} [c_2 c_3 c_4 ...c_k c_1G]=0
$$
since $a_jG=0$ and (\ref{eq:Wick-induction-new}) follows from (\ref{eq:Wick-induction-1}).

Finally if $c_1=a_j^*$ for some $j\notin I$ then 
\[
{\mathop{\rm Tr}\nolimits} [c_1 c_2 c_3 c_4 ...c_k G]={\mathop{\rm Tr}\nolimits} [c_2 c_3 c_4 ...c_k Gc_1]=0
\]
since $G a_j^*=0$ and we obtain (\ref{eq:Wick-induction-new}).
\end{proof}

\addcontentsline{toc}{section}{Acknowledgments}
\text{}\\
{\bf Acknowledgments:} I thank my advisor Jan Philip Solovej for giving me the problem and various helpful discussions.


\end{document}